\documentclass[journal]{IEEEtran}

\usepackage{subfigure}
\usepackage{setspace}
\usepackage{multirow} 
\usepackage{algorithm}
\usepackage{algorithmic}
\usepackage{amsmath}
\usepackage{amssymb}
\usepackage{latexsym}
\usepackage{multirow}
\usepackage{epsfig}
\usepackage{graphics}
\usepackage{graphicx}
\usepackage{mathrsfs}
\usepackage{subfigure}
\usepackage{slashbox}
\usepackage{mathrsfs}
\usepackage{bbding}
\usepackage{cite}
\usepackage{color}
\usepackage{xcolor}
\usepackage{booktabs}
\usepackage{amssymb,mathrsfs,amsmath}
\newcommand{\bm}[1]{\mbox{\boldmath{$#1$}}}
\usepackage{cases}

\usepackage{ragged2e}

\usepackage{hyperref}

\usepackage{amsthm}
\theoremstyle{plain}
\newtheorem{prop}{Proposition}

\newtheorem{rem}{Remark}

\allowdisplaybreaks [4]

\begin{document}
	
\title{Exploiting Intelligent Reflecting Surfaces for Interference Channels with SWIPT} 
\author{Ying~Gao, Qingqing~Wu,~\IEEEmembership{Senior Member,~IEEE,} Wen~Chen,~\IEEEmembership{Senior Member,~IEEE,} Celimuge~Wu,~\IEEEmembership{Senior Member,~IEEE,} Derrick~Wing~Kwan~Ng,~\IEEEmembership{Fellow,~IEEE}, and Naofal~Al-Dhahir,~\IEEEmembership{Fellow,~IEEE}
\thanks{The work of Q.~Wu was supported by NSFC 62371289 and NSFC 62331022. The work of W.~Chen was supported by National Key Project 2020YFB1807700, NSFC 62071296, Shanghai 22JC1404000, 20JC1416502, and PKX2021-D02. The work of C.~Wu was supported in part by the ROIS NII Open Collaborative Research 23S0601, and in part by JSPS KAKENHI grant number 21H03424. The work of D.~W.~K.~Ng was supported by the Australian Research Council's Discovery Projects (DP210102169, DP230100603). The work of N.~Al-Dhahir was supported by Erik Jonsson Distinguished Professorship at UT-Dallas. The associate editor coordinating the review of this article and approving it for publication was A.S. Cacciapuoti. \emph{(Corresponding author: Qingqing Wu.)} }
\thanks{Y.~Gao is with the Department of Electronic Engineering, Shanghai Jiao Tong University, Shanghai 201210, China, and also with the State Key Laboratory of Internet of Things for Smart City, University of Macau, Macao 999078, China (e-mail: yinggao@um.edu.mo).}
\thanks{Q.~Wu and W.~Chen are with the Department of Electronic Engineering, Shanghai Jiao Tong University, Shanghai 201210, China (e-mail: qingqingwu@sjtu.edu.cn; whenchen@sjtu.edu.cn).}
\thanks{C.~Wu is with the Graduate School of Informatics and Engineering, The University of Electro-Communications, Tokyo 182-8585, Japan
(e-mail: clmg@is.uec.ac.jp).}
\thanks{D.~W.~K.~Ng is with the School of Electrical  Engineering and Telecommunications, University of New South Wales, NSW 2052, Australia (e-mail: w.k.ng@unsw.edu.au).}
\thanks{N.~Al-Dhahir is with the Department of Electrical and Computer Engineering, The University of Texas at Dallas, Richardson, TX 75080 USA (e-mail: aldhahir@utdallas.edu).}}

\maketitle

\begin{abstract}
	This paper considers intelligent reflecting surface (IRS)-aided simultaneous wireless information and power transfer (SWIPT) in a multi-user multiple-input single-output (MISO) interference channel (IFC), where multiple transmitters (Txs) serve their corresponding receivers (Rxs) in a shared spectrum with the aid of IRSs. Our goal is to maximize the sum rate of the Rxs by jointly optimizing the transmit covariance matrices at the Txs, the phase shifts at the IRSs, and the resource allocation subject to the individual energy harvesting (EH) constraints at the Rxs. Towards this goal and based on the well-known power splitting (PS) and time switching (TS) receiver structures, we consider three practical transmission schemes, namely the IRS-aided hybrid TS-PS scheme, the IRS-aided time-division multiple access (TDMA) scheme, and the IRS-aided TDMA-D scheme. The latter two schemes differ in whether the Txs employ deterministic energy signals known to all the Rxs. Despite the non-convexity of the three optimization problems corresponding to the three transmission schemes, we develop computationally efficient algorithms to address them suboptimally, respectively, by capitalizing on the techniques of alternating optimization (AO) and successive convex approximation (SCA). Moreover, we conceive feasibility checking methods for these problems, based on which the initial points for the proposed algorithms are constructed. Simulation results demonstrate that our proposed IRS-aided schemes significantly outperform their counterparts without IRSs in terms of sum rate and maximum EH requirements that can be satisfied under various setups. In addition, the IRS-aided hybrid TS-PS scheme generally achieves the best sum rate performance among the three proposed IRS-aided schemes, and if not, increasing the number of IRS elements can always accomplish it. 
\end{abstract}


\section{Introduction}	
Harvesting energy from the environment, which can provide cost-effective and almost unlimited energy supplies for Internet-of-Things (IoT) devices, is a more convenient and greener alternative to conventional battery-powered solutions \cite{2011_Sujesha_EH_survey}. Compared to other renewable energy sources (e.g., solar and wind), radio-frequency (RF) signals are more appealing for energy harvesting (EH) due to their high availability and controllability. On the other hand, as RF signals carry both information and energy, the notion of \emph{simultaneous wireless information and power transfer} (SWIPT) comes into the picture and has attracted increasing attention. Indeed, the SWIPT concept was first proposed by Varshney, who also derived a capacity-energy function to characterize the rate-energy (R-E) trade-off for a single-input single-output (SISO) flat fading channel in \cite{2008_Varshney_SWIPT}. Besides, the authors of \cite{2010_Pulkit_SWIPT} extended the work in \cite{2008_Varshney_SWIPT} to a SISO frequency-selective channel and revealed that a non-trivial R-E trade-off exists in frequency-domain power allocation. Both of the above works assumed that the receiver circuits can decode information and harvest energy from the same RF signal concurrently, which, however, may not be realizable in practice yet due to technical limitations in circuitries \cite{2013_Rui_MIMO_SWIPT}. Hence, the authors of \cite{2013_Rui_MIMO_SWIPT} proposed two practical receiver structures, namely, time switching (TS) and power splitting (PS), for co-located receivers with both information decoding (ID) and EH requirements. Moreover, the R-E regions of three-node multiple-input multiple-output (MIMO) broadcast channels were studied in \cite{2013_Rui_MIMO_SWIPT}, where both cases of separated and co-located receivers were considered. Also, following \cite{2013_Rui_MIMO_SWIPT}, the R-E trade-off study was extended to multi-user multiple-input single-output (MISO) broadcast systems with TS receivers \cite{2015_Yanjie_TS}, PS receivers \cite{2014_Qingjiang_PS}, coexisting PS and TS receivers \cite{2019_Ruihong_coexist}, and separated receivers \cite{2014_Jie_SWIPT}. 

Meanwhile, there have been several SWIPT studies for more general multi-point-to-multi-point systems, where multiple transmitters communicate
with their corresponding receivers over a commonly shared spectrum, e.g., \cite{2013_Park_two-user,2014_Park_K-user,2016_Jian_robust_IFC,2014_Baosheng_SWIPT,2014_Qingjiang_SOCP,2014_Chao_IFC,2016_Ming-Min_SWIPT,2015_Zhiyuan_SWIPT}. From the perspective of wireless information transfer (WIT), such systems are generally modeled as interference channels (IFCs). However, unlike conventional IFC without EH, the cross-link interference, while generally detrimental for ID, is helpful for EH since strong interference serves as an abundant energy source. Therefore, interesting R-E trade-offs naturally exist in SWIPT IFCs and the design of transmission strategies for SWIPT is a critical issue in IFCs. Depending on whether the information and energy receivers are geographically separated or co-located, we can identify two different lines of research along this direction. For the case of separated receivers, the authors of \cite{2013_Park_two-user} first investigated SWIPT for a two-user MIMO IFC and then extended their work to a more general $K$-user MIMO IFC in \cite{2014_Park_K-user}, where the necessary condition for the optimal transmission strategy was revealed. Besides, the study in \cite{2013_Park_two-user} was also extended to the case of imperfect channel state information (CSI) at the transmitters in \cite{2016_Jian_robust_IFC}. On the other hand, for the case of co-located receivers, the authors of \cite{2014_Baosheng_SWIPT} investigated the joint design of transmit power allocation and receive PS ratios for a two-user SISO IFC with SWIPT. This work was then extended to $K$-user MISO and MIMO IFCs in \cite{2014_Qingjiang_SOCP,2014_Chao_IFC,2016_Ming-Min_SWIPT} and \cite{2015_Zhiyuan_SWIPT}, respectively. In particular, the authors of \cite{2014_Chao_IFC} proposed three practical SWIPT schemes, namely, the PS scheme, the time-division mode switching (TDMS) scheme, and the time-division multiple access (TDMA) scheme. Interestingly, it was shown in \cite{2014_Chao_IFC} that none of the three schemes can consistently outperform the other two in terms of maximizing the sum rate of the receivers while satisfying their EH requirements. Specifically, the PS scheme generally performs the best when the direct links dominate, while the TDMS scheme generally performs the best when the cross links dominate. Additionally, if the interference is overwhelming and the EH requirement is also stringent, the TDMA scheme may perform the best. 

Despite the fruitful progress in the development of theory, there are still various challenges in implementing SWIPT systems in practice. For example, due to severe path loss, the efficiency of wireless power transfer (WPT) decreases rapidly with the increase of transmission distance, which fundamentally limits the performance of SWIPT systems. Recently, intelligent reflecting surface (IRS), composed of a large number of low-cost reflecting elements, has emerged as a promising solution to boost the efficiencies of both WPT and WIT \cite{2022_Qingqing_WEIT_overview}. By adapting each reflecting element in real-time to tune the phase shifts of the impinging signals, an IRS can achieve fine-grained passive beamforming gains, thereby reconfiguring the wireless propagation channels to achieve different desired design objectives \cite{2020_Qingqing_Towards}. Particularly, it was first revealed in \cite{2019_Qingqing_Joint} that an IRS is capable of providing an asymptotic squared power gain in terms of the user receive power, which has prompted considerable research to exploit IRSs for enhancing the performance of WIT \cite{2019_Chongwen_EE,2019_Miao_Secure,2021_Meng_CoMP,2022_Kaitao_Sensing}, WPT \cite{2022_Chi_WPT}, both downlink WPT and uplink WIT \cite{2022_Qingqing_WPCN,2022_Guangji_MEC}, and SWIPT \cite{2020_Qingqing_SWIPT_letter,2020_Qingqing_SWIPT_QoS,2020_Cunhua_SWIPT,2021_Shayan_SWIPT,2022_Yang_SWIPT}. For instance, for an IRS-aided SWIPT system, it was shown in \cite{2020_Qingqing_SWIPT_letter} that the operating range of WPT and the corresponding achievable signal-to-interference-plus-noise ratio (SINR)-energy region can be significantly enlarged by deploying an IRS in close proximity to the energy receivers. Also, the results in \cite{2020_Qingqing_SWIPT_QoS} and \cite{2020_Cunhua_SWIPT} indicated that IRS-aided SWIPT systems are clearly superior to their conventional counterparts without the IRS in terms of the total transmit power required at the access point and the weighted sum rate of the information receivers, respectively. Besides \cite{2020_Qingqing_SWIPT_letter,2020_Qingqing_SWIPT_QoS,2020_Cunhua_SWIPT} where the information and energy receivers are separated, the authors of \cite{2021_Shayan_SWIPT} considered co-located receivers based on the PS receiver structure and studied the max-min energy efficiency (EE) problem for an IRS-aided SWIPT system. Their simulation results verified the effectiveness of the IRS in enhancing the system EE. In addition, the authors of \cite{2022_Yang_SWIPT} showed that the IRS can introduce remarkable R-E gains with robustness to quantized IRS phase shifts. 

Although the above works have validated the advantages of integrating IRSs into SWIPT systems under various setups and criteria, to our best knowledge, the research on IRS-aided SWIPT in IFCs is still in its infancy. We note that there have been several studies on IFCs aided by IRSs (see e.g., \cite{2020_Wei_IFC,2022_Abdollahi_IFC,2022_Tao_IFC}). However, the transmit/reflect beamforming and resource allocation designs proposed in \cite{2020_Wei_IFC,2022_Abdollahi_IFC,2022_Tao_IFC} do not consider WPT. Their results may not achieve satisfactory R-E trade-offs in IRS-aided IFCs with SWIPT and even cannot satisfy the EH requirements at the receivers. Moreover, some fundamental issues remain unknown or uninvestigated when it comes to introducing IRSs into IFCs with SWIPT. First, \emph{will the performance comparison results among the three transmission strategies in \cite{2014_Chao_IFC} still hold after the introduction of IRSs}? This question is motivated by the fact that compared with the PS scheme without time division, the TDMS and TDMA schemes can allow IRSs to adopt different phase-shift vectors in different time slots to customize favorable time-varying channels. Thus, it is unknown whether IRSs can bring more performance gains to the TDMS and TDMA schemes than that to the PS scheme, thereby rewriting the performance comparison results in the case without IRSs. Besides, since none of the three schemes in \cite{2014_Chao_IFC} can always outperform one another, another question arises: \emph{Can we propose a new transmission strategy that unifies some of these three schemes, or even always outperforms them}? 

\begin{figure}[!t]
	\centering
	\includegraphics[width=0.48\textwidth]{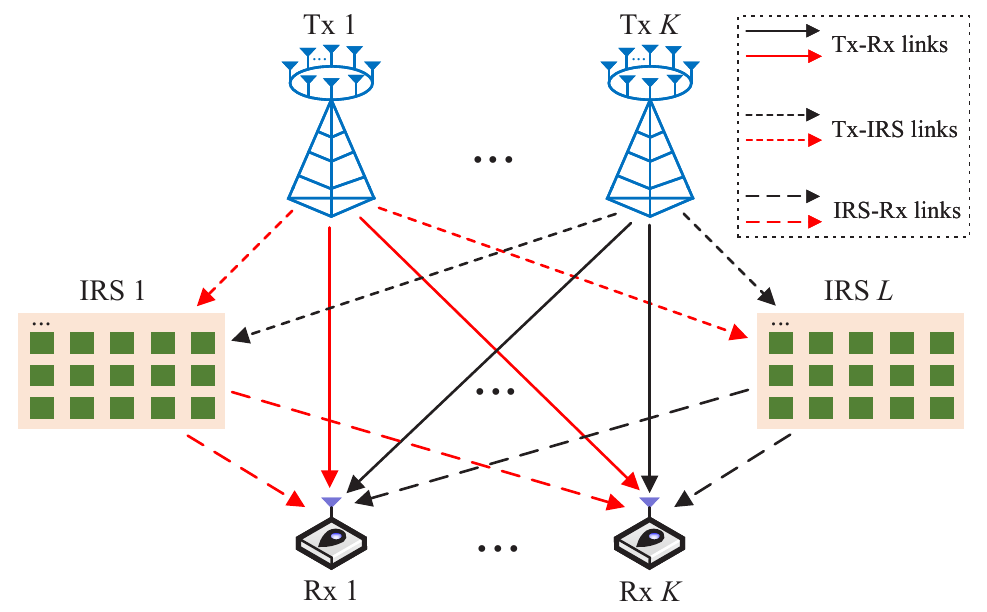}
	\caption{Illustration of IRS-aided SWIPT in a $K$-user MISO IFC.}
	\label{fig:system_model}
\end{figure}

Motivated by the aforementioned open issues, we study IRS-aided SWIPT in a multi-user MISO IFC, where multiple IRSs are deployed to assist the downlink SWIPT between multiple transmitters (Txs) and their corresponding receivers (Rxs), as shown in Fig. \ref{fig:system_model}. We aim to maximize the sum rate of the Rxs by jointly optimizing the transmit covariance matrices at the Txs, the phase shifts at the IRSs, and the resource allocation subject to the constraints of individual EH requirements at the Rxs. Our main contributions are summarized as follows. 
\begin{itemize}
	\item We consider three practical transmission schemes for SWIPT in IFCs, namely the IRS-aided hybrid TS-PS scheme, the IRS-aided TDMA scheme, and the IRS-aided TDMA-D scheme. The first scheme is a new strategy that unifies the PS and TDMS schemes in \cite{2014_Chao_IFC}. The third scheme differs from the second one in that the energy signals transmitted by the Txs are known deterministic signals to all the Rxs. For each of the schemes, we formulate a sum rate maximization problem, where the transmit covariance matrices and IRS phase shifts in different time slots are jointly optimized with the resource allocation to fully utilize the available degrees-of-freedom (DoF). These non-convex problems, characterized by coupled optimization variables, are considerably distinct and significantly more challenging compared to those in \cite{2014_Chao_IFC}. This is mainly because of the newly introduced time-varying IRS phase-shift variables. Additionally, the unified hybrid TS-PS scheme involves the optimization of both the TS- and PS-related variables, which further complicates the solution development.   
	\item To handle these non-convex formulated problems, we first propose suboptimal algorithms based on the proper change of variables, the alternating optimization (AO) method, and the successive convex approximation (SCA) technique. Moreover, useful properties of the obtained solutions are revealed to offer important engineering insights. Finally, we provide methods for checking the feasibility of these formulated problems, based on which we construct initial points for the proposed algorithms. 
	\item Our numerical results demonstrate that: 
	1) the performance advantages of the distributed IRS deployment over the centralized IRS deployment or vice verse depend on the  geographical distribution of the Tx-Rx pairs; 2) the introduction of IRSs rewrites some fundamental results in the comparisons of the considered schemes; 3) our proposed IRS-aided schemes perform much better in achievable rate and harvested energy than their counterparts without IRSs; 4) the IRS-aided hybrid TS-PS scheme is generally the best among all the considered schemes for obtaining the highest sum rate, and if not, we can always increase the number of IRS elements to achieve this. 
\end{itemize} 

The rest of this paper is organized as follows. Section \ref{Sec_model_formu} introduces the system model and problem formulations for a multi-user MISO IFC under different transmission strategies. Algorithms proposed for the three formulated problems are presented in Sections \ref{Sec_P1_solution} and \ref{Sec_P2_solution}. In Section \ref{Sec_feasibility}, we present feasibility checking and initialization methods for the formulated problems and the proposed algorithms, respectively. Section \ref{Sec_simulation} provides numerical results to demonstrate the effectiveness of our proposed algorithms and compare their performance. Section \ref{Sec_conclusion} concludes the paper. 

\emph{Notations:} Lower-case letters refer to scalars, while boldface lower-case (upper-case) letters represent vectors (matrices). Let $\mathbb C^{M\times N}$ denote the space of $M\times N$ complex-valued matrices. The trace, transpose, conjugate transpose, and expectation operators are denoted by ${\rm tr}(\cdot)$, $(\cdot)^T$, $(\cdot)^H$, and $\mathbb E\left(\cdot\right)$, respectively. ${\rm diag}\left(\cdot\right)$ represents the diagonalization operation. We denote the cardinality of a set $\mathcal K$ by $\left|\mathcal K\right|$. $\mathcal{CN}\left(\bm\mu,\mathbf \Sigma\right)$ denotes a complex Gaussian distribution with a mean vector $\bm\mu$ and co-variance matrix $\mathbf \Sigma$. $\{a_i\}_{i\in\mathcal I}$ represents the collection of all variables $a_i$, where $i$ belongs to the set $\mathcal I$. $\mathbf 0$ is an all-zero matrix whose dimension is determined by the context. $\mathbf S\succeq \mathbf 0$ indicates that $\mathbf S$ is a positive semidefinite matrix. $\jmath$ refers to the imaginary unit, i.e., $\jmath^2 = 1$, and ${\rm Re}\{\cdot\}$ denotes the real part of a complex number. $[\cdot]_m$ refers to the $m$-th element of a vector, while the $(m,n)$-th element of a matrix is denoted by $[\cdot]_{m,n}$. ${\rm rank}(\cdot)$ represents the rank of a matrix. 

\section{System Model and Problem Formulation}\label{Sec_model_formu}
As depicted in Fig. \ref{fig:system_model},  we consider a narrow-band IRS-aided MISO IFC with $L$ passive IRSs\footnote{As an initial study on performance comparison among different IRS-aided transmission strategies for SWIPT in IFCs, we consider the simplest type of IRS architecture, i.e., passive IRS \cite{2021_Qingqing_Tutorial}. It will be interesting to extend this work to active IRS \cite{2023_Ying_Active_SWIPT}, hybrid IRS \cite{2022_Nguyen_hybrid}, intelligent omni-surface \cite{2022_Shuhang_omni}, etc., to see whether the obtained results still hold, which are left for future work. \looseness=-1} and $K$ Tx-Rx pairs, indexed by $\mathcal L = \{1,\cdots,L\}$ and $\mathcal K = \{1,\cdots,K\}$, respectively. All the Txs are assumed to operate over the same frequency band. Each $M$-antenna Tx serves its corresponding single-antenna Rx and simultaneously interferes with the other $K-1$ Rxs. Particularly, in the presence of the $L$ IRSs, the communication links from Tx $k$ to Rx $k$ ($k \in\mathcal K$) include not only a direct link but also reflected links. Note that due to the multiplicative path loss, we ignore the signals undergoing two or more reflections \cite{2020_Cunhua_SWIPT,2021_Shayan_SWIPT,2022_Yang_SWIPT}. At each Rx, the received signal can be exploited for either ID or EH, or be split into two parts for ID and EH, respectively. Without loss of generality, the overall transmission interval is normalized to unity. 

Suppose that IRS $\ell$ ($\ell \in\mathcal L$) is equipped with $N_\ell$ reflecting elements. The set of all the IRS elements is denoted as $\mathcal N$ with $\left|\mathcal N\right| = N = \sum_{\ell=1}^L N_\ell$. Let $\mathbf G_{i,\ell}\in\mathbb C^{N_\ell \times M}$, $\mathbf f_{\ell,k}^H \in\mathbb C^{1\times N_\ell}$, and $\mathbf h_{i,k}^H \in\mathbb C^{1\times M}$ denote the equivalent baseband channels from Tx $i$ to IRS $\ell$, from IRS $\ell$ to Rx $k$, and from Tx $i$ to Rx $k$, respectively. Denoted by $\mathbf \Phi_{i,\ell,k} = {\rm diag}\left(\mathbf f_{\ell,k}^H\right)\mathbf G_{i,\ell} \in\mathbb C^{N_{\ell}\times M}$ the cascaded channel from Tx $i$ to Rx $k$ via IRS $\ell$. To quantify the performance limit, we assume that perfect CSI of all the channels involved can be acquired by utilizing existing channel estimation methods, e.g., \cite{2021_Wei_channel,2021_Wei_channel_conf}. 

\begin{figure}[!t]
	\centering
	\subfigure[]{\label{fig:hybrid_general_model}
		\includegraphics[width=0.37\textwidth]{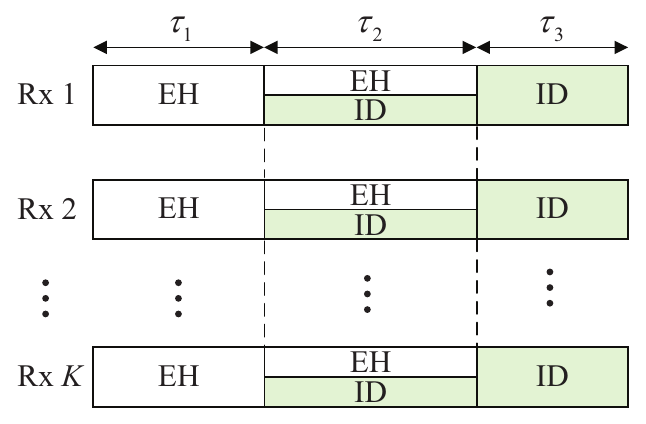}}
	\subfigure[]{\label{fig:TDMA_model}
		\includegraphics[width=0.37\textwidth]{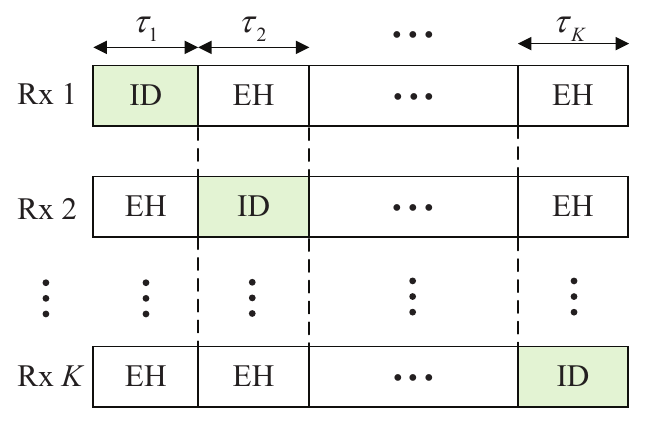}}
	\caption{Illustration of the proposed transmission strategies: (a) IRS-aided hybrid TS-PS scheme; (b) IRS-aided TDMA and TDMA-D schemes.}
\end{figure}

\subsection{IRS-Aided Hybrid TS-PS Scheme}
For the hybrid TS-PS scheme shown in Fig. \ref{fig:hybrid_general_model}, the transmission interval is divided into three time slots. All the Rxs operate in the EH mode in time slot $1$, split the received signal into two streams for ID and EH, respectively, in time slot $2$, and switch to the ID mode in time slot $3$. The time fraction of the $j$-th time slot is assumed to be $\tau_j \in\left[0,1\right] $, $j\in\{1,2,3\}$. Denoted by $\bm x_{i,j} \in\mathbb C^{M\times 1}$ the signal vector transmitted by Tx $i$ in time slot $j$. It is assumed that $\{\bm x_{i,j}\}_{i\in\mathcal K}$ are independent over $i$. We further assume that a Gaussian codebook is adopted at each Tx and $\bm x_{i,j} \sim\mathcal {CN}\left(\mathbf 0, \mathbf S_{i,j}\right)$, where $\mathbf S_{i,j} = \mathbb E\left(\bm x_{i,j}\bm x_{i,j}^H\right) \succeq \mathbf 0$ denotes the transmit covariance matrix for Tx $i$ in time slot $j$ \cite{2014_Chao_IFC}. Let $\mathbf \Theta_{\ell,j} = {\rm diag}\left(\beta_{\ell,1,j}e^{\jmath\theta_{\ell,1,j}}, \cdots, \beta_{\ell,N_\ell,j}e^{\jmath\theta_{\ell,N_\ell,j}} \right)$ denote the reflection-coefficient matrix exploited at the $\ell$-th IRS in time slot $j$, where $\beta_{\ell,n,j}$ and $\theta_{\ell,n,j}$ represent the reflection amplitude and phase shift of element $n \in \{1,\cdots,N_\ell\}$, respectively. In addition, $\beta_{\ell,n,j}$ and $\theta_{\ell,n,j}$ can be independently adjusted over $[0,1]$ and $[0,2\pi)$, respectively \cite{2021_Qingqing_Tutorial,2016_Huanhuan_amplitude}. It is not difficult to see that in the case of without IRSs, this scheme reduces to the PS and TDMS schemes in \cite{2014_Chao_IFC} when $\tau_1 = \tau_3 = 0$ and $\tau_2 = 0$, respectively.  
\begin{itemize}
	\item Time slot $1$: The baseband received signal at Rx $k$, $k\in\mathcal K$, in time slot $1$ is given by
	\begin{align}\label{Hy_received_signal_1}	
	y_{k,1} & = \sum_{i=1}^K\left(\sum_{\ell=1}^L\mathbf f_{\ell,k}^H\mathbf \Theta_{\ell,1}\mathbf G_{i,\ell} + \mathbf h_{i,k}^H \right)\bm x_{i,1} + \tilde{n}_k \nonumber\\
	& = \sum_{i=1}^K\left(\sum_{\ell=1}^L\mathbf u_{\ell,1}^H\mathbf \Phi_{i,\ell,k} + \mathbf h_{i,k}^H \right)\bm x_{i,1} + \tilde{n}_k \nonumber\\
	& = \sum_{i=1}^K\mathbf v_1^H\mathbf H_{i,k}\bm x_{i,1} + \tilde{n}_k,
	\end{align}
	where $\mathbf u_{\ell,1}^H = \left[\beta_{\ell,1,1}e^{\jmath\theta_{\ell,1,1}}, \cdots, \beta_{\ell,N_\ell,1}e^{\jmath\theta_{\ell, N_{\ell},1}}\right]$, $\mathbf v_1 = \left[\mathbf u_{1,1}^H,\cdots,\mathbf u_{L,1}^H, 1\right]^H$ denotes the phase-shift vector associated with all the IRSs in time slot $1$, $\mathbf H_{i,k} = \left[\mathbf \Phi_{i,1,k}; \cdots; \mathbf\Phi_{i,L,k}; \mathbf h_{i,k}^H\right]$, and $\tilde n_k \sim\mathcal {CN} \left(0, \tilde \sigma_k^2 \right)$ denotes the RF-band antenna noise at Rx $k$ with $\sigma_k^2$ being the noise variance. By adopting the linear EH model \cite{2020_Qingqing_SWIPT_letter,2020_Qingqing_SWIPT_QoS,2020_Cunhua_SWIPT} and ignoring the negligible noise power, the harvested RF-band energy at Rx $k$ during $\tau_1$ can be expressed as $Q_{k,1} = \zeta\tau_1\sum_{i = 1}^K{\rm tr}\left(\mathbf a_{i,k,1}\mathbf a_{i,k,1}^H\mathbf S_{i,1}\right)$,  where $\mathbf a_{i,k,1}^H = \mathbf v_1^H\mathbf H_{i,k}$ and $\zeta\in\left(0,1\right] $ represents the constant energy conversion efficiency of each Rx.  
	\item Time slot $2$: At Rx $k$, the portion of signal power split for ID is denoted by $\rho_k \in \left[0,1\right] $ and that for EH by $1 - \rho_k$. Similar to \eqref{Hy_received_signal_1}, the received signal at Rx $k$ in time slot $2$ for ID and EH can be written as $y_{k,2}^{\rm ID} = \sqrt{\rho_k}\left(\sum_{i=1}^K\mathbf v_2^H\mathbf H_{i,k}\bm x_{i,2} + \tilde n_k\right) + \hat n_k$ and $y_{k,2}^{\rm EH} = \sqrt{1-\rho_k}\left(\sum_{i=1}^K\mathbf v_2^H\mathbf H_{i,k}\bm x_{i,2} + \tilde n_k\right)$, respectively, where $\mathbf v_2$ denotes the phase-shift vector associated with all the IRSs in this time slot with $\mathbf v_2 = \left[\mathbf u_{1,2}^H,\cdots,\mathbf u_{L,2}^H, 1\right]^H$ and $\mathbf u_{\ell,2}^H = \left[\beta_{\ell,1,2}e^{\jmath\theta_{\ell,1,2}}, \cdots, \beta_{\ell,N_\ell,2}e^{\jmath\theta_{\ell, N_{\ell},2}}\right]$, and $\hat n_k \sim\mathcal {CN} \left(0, \hat \sigma_k^2 \right)$ represents the baseband processing noise. By treating the cross-link interference as noise, the achievable rate at Rx $k$ in bits/second/Hertz (bps/Hz) is given by $R_{k,2} = \tau_2\log_2\left(1 + \gamma_{k,2}\right)$ with
	\begin{align}\label{Hy_SINR_2}
	\gamma_{k,2}  = \frac{{\rm tr}\left(\mathbf a_{k,k,2}\mathbf a_{k,k,2}^H \mathbf S_{k,2}\right) }{\sum_{i=1,i\neq k}^K{\rm tr}\left(\mathbf a_{i,k,2}\mathbf a_{i,k,2}^H \mathbf S_{i,2}\right) + \tilde\sigma_k^2 + \frac{\hat\sigma_k^2}{\rho_k}},
	\end{align}
	where $\mathbf a_{i,k,2}^H = \mathbf v_2^H\mathbf H_{i,k}$, $i,k\in\mathcal K$. On the other hand, the harvested RF-band energy at Rx $k$ during $\tau_2$ can be expressed as $Q_{k,2} = \zeta\tau_2\left(1 - \rho_k\right)\sum_{i = 1}^K{\rm tr}\left(\mathbf a_{i,k,2}\mathbf a_{i,k,2}^H\mathbf S_{i,2}\right)$, $k\in\mathcal K$. 
    \item Time slot $3$: In this time slot, the achievable rate at Rx $k$, $k\in\mathcal K$, in bps/Hz can be expressed as $R_{k,3} = \tau_3\log_2\left(1 + \gamma_{k,3}\right)$ with 
    \begin{align}\label{Hy_SINR_3}
    \gamma_{k,3}  = \frac{{\rm tr}\left(\mathbf a_{k,k,3}\mathbf a_{k,k,3}^H \mathbf S_{k,3}\right) }{\sum_{i=1,i\neq k}^K{\rm tr}\left(\mathbf a_{i,k,3}\mathbf a_{i,k,3}^H \mathbf S_{i,3}\right) + \sigma_k^2},
    \end{align}
    where $\sigma_k^2 = \tilde{\sigma}_k^2 + \hat{\sigma}_k^2$, which means that two types of noise exist at Rx $k$: the RF-band antenna noise $\tilde n_k$ and the baseband processing noise $\hat n_k$ \cite{2013_Rui_MIMO_SWIPT}. Additionally, $\mathbf a_{i,k,3}^H = \mathbf v_3^H\mathbf H_{i,k}$ with $\mathbf v_3$ denoting the phase-shift vector associated with all the IRSs in time slot $3$, $\mathbf v_3 = \left[\mathbf u_{1,3}^H,\cdots,\mathbf u_{L,3}^H, 1\right]^H$ and $\mathbf u_{\ell,3}^H = \left[\beta_{\ell,1,3}e^{\jmath\theta_{\ell,1,3}}, \cdots, \beta_{\ell,N_\ell,3}e^{\jmath\theta_{\ell, N_{\ell},3}}\right]$. 
\end{itemize}

\subsection{IRS-Aided TDMA and TDMA-D Schemes}
For the IRS-aided TDMA-based schemes shown in Fig. \ref{fig:TDMA_model}, we divide the transmission interval into $K$ time slots, each of which occupies $\tau_j \in\left[0,1\right]$ fraction of the time, $\forall j\in\mathcal K$. The $k$-th Rx operates in the ID mode only in the $k$-th time slot while in the EH mode in other $K-1$ time slots. Let $\mathbf S_{i,j}\succeq \mathbf 0$ denote the transmit covariance matrix for the $i$-th Tx in the $j$-th time slot. In addition, $\mathbf v_j = \left[ \beta_{1,1,j}e^{\jmath\theta_{1,1,j}}, \cdots, \beta_{1, N_1,j}e^{\jmath\theta_{1, N_1,j}}, \cdots,\beta_{L,1,j}e^{\jmath\theta_{L,1,j}}, \cdots,\right. $\\$\left. \beta_{L,N_L,j}e^{\jmath\theta_{L, N_L,j}},1\right]^H$ denotes the phase-shift vector related to the $L$ IRSs in time slot $j$. 
\subsubsection{IRS-Aided TDMA Scheme}
Assume that in time slot $k$, the interference at Rx $k$ induced by the random unknown energy signals from Tx $i$, $i\in\mathcal K\backslash\{k\}$, cannot be canceled. Then, the achievable rate at Rx $k$ in bps/Hz is given by $R_k = \tau_k\log_2\left(1 + \gamma_k\right)$ with 
\begin{align}
\gamma_k  = \frac{{\rm tr}\left(\mathbf b_{k,k}\mathbf b_{k,k}^H \mathbf S_{k,k}\right) }{\sum_{i=1,i\neq k}^K{\rm tr}\left(\mathbf b_{i,k}\mathbf b_{i,k}^H \mathbf S_{i,k}\right) + \sigma_k^2},
\end{align}
where $\mathbf b_{i,k}^H = \mathbf v_k^H\mathbf H_{i,k}$, $i,k\in\mathcal K$. On the other hand, the harvested energy at Rx $k$ over the whole transmission interval is given by $Q_k = \sum_{j=1,j\neq k}^K\zeta\tau_j\sum_{i=1}^K{\rm tr}\left(\mathbf c_{i,k,j}\mathbf c_{i,k,j}^H\mathbf S_{i,j}\right)$, where $\mathbf c_{i,k,j}^H = \mathbf v_j^H\mathbf H_{i,k}$, $i,k,j\in\mathcal K$. 

\subsubsection{IRS-Aided TDMA-D Scheme}
Particularly, when a certain Rx operates in the EH mode, its corresponding Tx can transmit deterministic signals (e.g., training/pilot signals) known to all the Rxs. This means that in time slot $k$, Rx $k$ operating in the ID mode can cancel the cross-link interference introduced by the energy signals to improve its data rate. Accordingly, the achievable rate at Rx $k$, $k\in\mathcal K$, in bps/Hz can be expressed as $R_k^{\rm D} = \tau_k\log_2\left(1 + \gamma_k^{\rm D}\right)$, where 
\begin{align}
\gamma_k^{\rm D} = \frac{{\rm tr}\left(\mathbf b_{k,k}\mathbf b_{k,k}^H \mathbf S_{k,k}\right)}{\sigma_k^2}.
\end{align}

\subsection{Problem Formulation}
Our target is to maximize the sum rate of the Rxs by jointly optimizing the transmit covariance matrices and IRS phase-shift vectors over all the time slots together with the resource allocation while satisfying the individual EH requirements at the Rxs, given by $E_k$, $k\in\mathcal K$. For the IRS-aided hybrid TS-PS scheme, the problem of interest can be mathematically formulated as 
\begin{subequations}
	\begin{align}
	\text{(P1)}: \hspace{3mm}&\underset{\substack{\{\mathbf S_{i,j} \succeq \mathbf 0\}, \{\tau_j\}, \{\rho_k\},\\ \{\mathbf v_j\}, i,k\in\mathcal K, j\in\{1,2,3\}}}{\max} \hspace{2mm} \sum_{k=1}^K\Big[ \tau_2\log_2\left(1 + \gamma_{k,2} \right) \nonumber\\
	& \hspace{3cm} + \tau_3\log_2\left( 1 + \gamma_{k,3}\right)\Big] \\
	\text{s.t.} \hspace{3mm}&  \zeta\tau_1\sum_{i = 1}^K{\rm tr}\left(\mathbf a_{i,k,1}\mathbf a_{i,k,1}^H\mathbf S_{i,1}\right) \nonumber\\ 
	& + \zeta\tau_2\left(1 - \rho_k\right)\sum_{i = 1}^K{\rm tr}\left(\mathbf a_{i,k,2}\mathbf a_{i,k,2}^H\mathbf S_{i,2}\right) \geq E_k, \nonumber\\
	& \hspace{1.5mm}\forall k \in\mathcal K, \label{P1_cons:b}\\
	&  {\rm tr}\left(\mathbf S_{i,j} \right) \leq P_i, \ \forall i\in\mathcal K, j \in \{1,2,3\}, \label{P1_cons:c}\\
	&  \sum_{j=1}^3 \tau_j \leq 1, \  \tau_j \geq 0, \ \forall j \in\{1,2,3\}, \label{P1_cons:d}\\
	&  0 \leq \rho_k \leq 1, \ \forall k\in\mathcal K, \label{P1_cons:e}\\
	& \left|\left[\mathbf v_j\right]_n \right| \leq 1,  \ \left[\mathbf v_j\right]_{N+1} = 1, \ \forall n\in\mathcal N, j \in\{1,2,3\}, \label{P1_cons:f}
	\end{align}
\end{subequations} 
where $P_i > 0$ in \eqref{P1_cons:c} denotes the maximum instantaneous transmit power at Tx $i$, $i\in\mathcal K$, \eqref{P1_cons:d} and \eqref{P1_cons:e} impose the constraints on the time allocation and the PS ratios, respectively, and \eqref{P1_cons:f} ensures the modulus constraints on the IRS phase-shift vectors. Similarly, for the IRS-aided TDMA and TDMA-D schemes, the corresponding sum rate maximization problems can be respectively formulated as follows 
\begin{subequations}
	\begin{align}
	\text{(P2)}:\hspace{3mm} &\underset{\substack{\{\mathbf S_{i,j} \succeq \mathbf 0\}, \{\tau_j\},\\ \{\mathbf v_j\}, i,j\in\mathcal K}}{\max} \hspace{2mm} \sum_{k=1}^K\tau_k\log_2\left( 1 + \gamma_k\right) \\
	\text{s.t.}  \hspace{3mm} & \sum_{j=1,j\neq k}^K\zeta\tau_j\sum_{i=1}^K{\rm tr}\left(\mathbf c_{i,k,j}\mathbf c_{i,k,j}^H\mathbf S_{i,j}\right) \geq E_k, \ \forall k \in\mathcal K, \label{P2_cons:b}\\
	& {\rm tr}\left(\mathbf S_{i,j} \right) \leq P_i, \ \forall i,j\in\mathcal K, \label{P2_cons:c}\\
	& \sum_{j=1}^K \tau_j \leq 1,\  \tau_j \geq 0, \ \forall j \in \mathcal K, \label{P2_cons:d}\\
	& \left|\left[\mathbf v_j\right]_n \right| \leq 1, \ \left[\mathbf v_j\right]_{N+1} = 1, \ \forall n\in\mathcal N, j \in \mathcal K, \label{P2_cons:e}
	\end{align}
\end{subequations} 
\begin{subequations}
	\begin{align}
    \text{(P2-D)}: \underset{\substack{\{\mathbf S_{i,j} \succeq \mathbf 0\}, \{\tau_j\},\\ \{\mathbf v_j\}, i,j\in\mathcal K}}{\max} \hspace{3mm} & \sum_{k=1}^K\tau_k\log_2\left( 1 + \gamma_k^{\rm D}\right) \\
	\text{s.t.} \hspace{3mm} & \eqref{P2_cons:b} - \eqref{P2_cons:e}.
	\end{align} 
\end{subequations}
Note that none of the above three problems are convex due to the intricately coupled optimization variables in all the objective functions and the EH constraints. Therefore, it is very difficult, if not impossible, to solve these problems optimally. 

\section{Proposed Algorithm for (P1)}\label{Sec_P1_solution}
To tackle (P1), we resort to the AO method, which has been widely utilized in solving IRS-related optimization problems in existing works (e.g., \cite{2019_Chongwen_EE,2019_Miao_Secure,2020_Qingqing_SWIPT_letter}). Nevertheless, problem (P1) involves more sets of optimization variables than those considered in \cite{2019_Chongwen_EE,2019_Miao_Secure,2020_Qingqing_SWIPT_letter} that typically only focus on the transmit and reflect beamforming design. Despite this, we efficiently handle these optimization variables by dividing them into only two blocks, i.e., $\left\lbrace \{\mathbf S_{i,j}\}, \{\tau_j\}, \{\rho_k\}\right\rbrace$ and $\{\mathbf v_j\}$. Further details are provided in the next subsections.

\subsection{Optimizing $\left\lbrace \{\mathbf S_{i,j}\}, \{\tau_j\}, \{\rho_k\}\right\rbrace$ for Given $\{\mathbf v_j\}$} \label{Sec_P1_sub1}
For any fixed $\{\mathbf v_j\}$, 
(P1) is reduced to
\begin{subequations}\label{P1-sub1}
	\begin{align}
    \underset{\substack{\{\mathbf S_{i,j} \succeq \mathbf 0\}, \{\tau_j\}, \{\rho_k\},\\ \{\mathbf v_j\}, i,k\in\mathcal K, j\in\{1,2,3\}}}{\max} \hspace{3mm} & \sum_{k=1}^K\Big[ \tau_2\log_2\left(1 + \gamma_{k,2} \right) \nonumber\\
    & + \tau_3\log_2\left( 1 + \gamma_{k,3}\right)\Big] \\
	\text{s.t.} \hspace{3mm} & \eqref{P1_cons:b}-\eqref{P1_cons:e}. 
	\end{align}
\end{subequations} 
To make this problem more tractable, we apply the change of variables $\mathbf W_{i,j} = \tau_j\mathbf S_{i,j}$ and introduce slack variables $\{e_k\}$ such that $\frac{e_k}{\tau_2} = \frac{1}{\rho_k}$, where $i,k\in\mathcal K$ and $j\in\{1,2,3\}$. Then, problem \eqref{P1-sub1} can be equivalently transformed into problem \eqref{P1-sub1-Eqv1}, shown at the top of the next page,  
\begin{figure*}[!t]
\begin{subequations}\label{P1-sub1-Eqv1}
	\begin{align}
	& \hspace{-1.4cm}\underset{\substack{\{\mathbf W_{i,j} \succeq \mathbf 0\}, \{\tau_j\}, \{\rho_k\}, \\ \{e_k\}, i,k\in\mathcal K, j\in\{1,2,3\}}}{\max} \hspace{3mm}\sum_{k=1}^K\tau_2\log_2\left( 1 + \frac{\frac{{\rm tr}\left(\mathbf a_{k,k,2}\mathbf a_{k,k,2}^H \mathbf W_{k,2}\right)}{\tau_2} }{\frac{\sum_{i=1,i\neq k}^K{\rm tr}\left(\mathbf a_{i,k,2}\mathbf a_{i,k,2}^H \mathbf W_{i,2}\right)}{\tau_2} + \tilde\sigma_k^2 + \frac{e_k\hat\sigma_k^2}{\tau_2}}\right) \nonumber\\
	& \hspace{1.6cm} + \sum_{k=1}^K\tau_3\log_2\left( 1 + \frac{ \frac{{\rm tr}\left(\mathbf a_{k,k,3}\mathbf a_{k,k,3}^H \mathbf W_{k,3}\right)}{\tau_3} }{\frac{\sum_{i=1,i\neq k}^K{\rm tr}\left(\mathbf a_{i,k,3}\mathbf a_{i,k,3}^H \mathbf W_{i,3}\right)}{\tau_3} + \sigma_k^2}\right)  \label{P1-sub1-Eqv1_obj}\\
	\text{s.t.} \hspace{3mm} & \frac{e_k}{\tau_2} \geq \frac{1}{\rho_k}, \ \forall k\in\mathcal K, \label{P1-sub1-Eqv1_cons:b}\\
	& \zeta\sum_{i = 1}^K{\rm tr}\left(\mathbf a_{i,k,1}\mathbf a_{i,k,1}^H\mathbf W_{i,1}\right) + \zeta\left(1 - \rho_k\right)\sum_{i = 1}^K{\rm tr}\left(\mathbf a_{i,k,2}\mathbf a_{i,k,2}^H\mathbf W_{i,2}\right) \geq E_k, \ \forall k \in\mathcal K, \label{P1-sub1-Eqv1_cons:c}\\
	&  {\rm tr}\left(\mathbf W_{i,j}\right) \leq \tau_jP_i, \ \forall i\in\mathcal K, j\in\{1,2,3\}, \label{P1-sub1-Eqv1_cons:d}\\ 
	&  \eqref{P1_cons:d}, \eqref{P1_cons:e},
	\end{align}
\end{subequations}\hrulefill
\end{figure*}
where constraint \eqref{P1-sub1-Eqv1_cons:b} is obtained by replacing the equality signs in  $\frac{e_k}{\tau_2} = \frac{1}{\rho_k}$, $\forall k\in\mathcal K$, with inequality signs as it is active at the optimal solution to problem \eqref{P1-sub1-Eqv1},  since otherwise we can further improve the objective function by decreasing $e_k$. Next, we focus on solving the non-convex problem \eqref{P1-sub1-Eqv1}, whose non-convexity stems from the objective function and the constraints in \eqref{P1-sub1-Eqv1_cons:b} and \eqref{P1-sub1-Eqv1_cons:c}. Before tackling the non-concave objective function, we first rewrite it as 
$\sum_{k=1}^K\left(f_{k,2} - g_{k,2} + f_{k,3} - g_{k,3}\right) \triangleq R_{\rm sum}^{\rm Hy}$, where 
\begin{align}
& f_{k,2} = \tau_2\log_2\left(\frac{\sum\limits_{i=1}^K{\rm tr}\left(\mathbf a_{i,k,2}\mathbf a_{i,k,2}^H \mathbf W_{i,2}\right)}{\tau_2} + \frac{e_k\hat\sigma_k^2}{\tau_2} + \tilde\sigma_k^2 \right),\\
& g_{k,2} \nonumber\\
& = \tau_2\log_2\left(\frac{\sum\limits_{i=1,i\neq k}^K{\rm tr}\left(\mathbf a_{i,k,2}\mathbf a_{i,k,2}^H \mathbf W_{i,2}\right)}{\tau_2} + \frac{e_k\hat\sigma_k^2}{\tau_2} + \tilde\sigma_k^2 \right),\\
& f_{k,3} = \tau_3\log_2\left(\frac{\sum_{i=1}^K{\rm tr}\left(\mathbf a_{i,k,3}\mathbf a_{i,k,3}^H \mathbf W_{i,3}\right)}{\tau_3} + \sigma_k^2 \right),\\
& g_{k,3} = \tau_3\log_2\left(\frac{\sum_{i=1,i\neq k}^K{\rm tr}\left(\mathbf a_{i,k,3}\mathbf a_{i,k,3}^H \mathbf W_{i,3}\right)}{\tau_3} + \sigma_k^2 \right). 
\end{align}
It can be easily verified that the Hessian matrices of the above four functions are all negative semidefinite, i.e., $f_{k,2}$, $g_{k,2}$, $f_{k,3}$, and $g_{k,3}$ are all jointly concave with respect to (w.r.t.) their corresponding variables. Clearly, the concavity of $g_{k,2}$ and $g_{k,3}$ leads to the non-concavity of $R_{\rm sum}^{\rm Hy}$, which, however, paves the way for the application of the iterative SCA technique \cite{2020_Qingqing_SWIPT_letter}. Specifically, by using the fact that the first-order Taylor expansion of any concave function at any point is its global upper bound, we arrive at \eqref{Taylor_expansion_1} and \eqref{Taylor_expansion_2}, shown at the top of the next page,  
\begin{figure*}[!t]
\begin{align}
g_{k,2}\left(\mathbf W_2, e_k, \tau_2 \right) & \leq \tau_2^t\log_2(\Psi_{k,2}^t) + \frac{\sum\nolimits_{i=1,i\neq k}^K{\rm tr}\left(\mathbf a_{i,k,2}\mathbf a_{i,k,2}^H\left(\mathbf W_{i,2} - \mathbf W_{i,2}^t\right)\right)}{\Psi_{k,2}^t\ln2} 
+ \frac{\hat \sigma_k^2\left(e_k - e_k^t \right)}{\Psi_{k,2}^t\ln2} \nonumber \\
& + \left(\log_2(\Psi_{k,2}^t) - \frac{\Psi_{k,2}^t - \tilde\sigma_k^2}{\Psi_{k,2}^t\ln2}\right) \left(\tau_2 - \tau_2^t \right) \triangleq g_{k,2}^{\rm ub, \it t}\left(\mathbf W_2, e_k, \tau_2 \right), \ \forall k\in\mathcal K, \label{Taylor_expansion_1}\\
g_{k,3}\left(\mathbf W_3, \tau_3 \right) & \leq \tau_3^t\log_2(\Psi_{k,3}^t) + \frac{\sum\nolimits_{i=1,i\neq k}^K{\rm tr}\left(\mathbf a_{i,k,3}\mathbf a_{i,k,3}^H\left(\mathbf W_{i,3} - \mathbf W_{i,3}^t\right)\right)}{\Psi_{k,3}^t\ln2} \nonumber\\
& + \left(\log_2(\Psi_{k,3}^t) - \frac{\Psi_{k,3}^t - \sigma_k^2}{\Psi_{k,3}^t\ln2}\right)  \left(\tau_3 - \tau_3^t \right) \triangleq g_{k,3}^{\rm ub, \it t}\left(\mathbf W_3, \tau_3 \right), \  \forall k\in\mathcal K, \label{Taylor_expansion_2}
\end{align}\hrulefill
\end{figure*}
where $\mathbf W_j  = \{\mathbf W_{i,j}\}_{i\in\mathcal K\backslash\{k\}}$, $j \in \{2,3\}$, $\Psi_{k,2}^t = \frac{\sum_{i=1,i\neq k}^K{\rm tr}\left(\mathbf a_{i,k,2}\mathbf a_{i,k,2}^H \mathbf W_{i,2}^t\right)}{\tau_2^t} + \frac{e_k^t\hat\sigma_k^2}{\tau_2^t} + \tilde\sigma_k^2$, and $\Psi_{k,3}^t = \frac{\sum_{i=1,i\neq k}^K{\rm tr}\left(\mathbf a_{i,k,3}\mathbf a_{i,k,3}^H \mathbf W_{i,3}^t\right)}{\tau_3^t} + \sigma_k^2$. Besides, $\mathbf W_{i,j}^t$, $e_k^t$, $\rho_k^t$, and $\tau_j^t$ are the given local points in the $t$-th iteration. Accordingly, the objective function of problem \eqref{P1-sub1-Eqv1}, $R_{\rm sum}^{\rm Hy}$, is lower bounded by $\sum_{k=1}^K\left(f_{k,2} - g_{k,2}^{\rm ub, \it t}\left(\mathbf W_2, e_k, \tau_2 \right) + f_{k,3} - g_{k,3}^{\rm ub, \it t}\left(\mathbf W_3, \tau_3 \right)\right)$, which is a concave function.

To proceed, we deal with the non-convex constraints \eqref{P1-sub1-Eqv1_cons:b} and \eqref{P1-sub1-Eqv1_cons:c}. First, we note that \eqref{P1-sub1-Eqv1_cons:b} is equivalent to $\tau_2 \leq e_k\rho_k = \frac{1}{2}\left(e_k + \rho_k \right)^2 - \frac{1}{2}\left(e_k^2 + \rho_k^2 \right)$, $\forall k\in\mathcal K$. As the term $\frac{1}{2}\left(e_k + \rho_k \right)^2$ is jointly convex w.r.t. $e_k$ and $\rho_k$, it can be bounded from below by its first-order Taylor expansion at the given local points $e_k^t$ and $\rho_k^t$ in the $t$-th iteration, i.e., 
\begin{align}\label{Taylor_expansion_3}
\frac{1}{2}\left(e_k + \rho_k \right)^2 & \geq - \frac{1}{2}\left(e_k^t + \rho_k^t \right)^2 + \left(e_k^t + \rho_k^t \right)\left(e_k + \rho_k \right) \nonumber\\
& \triangleq \chi^{\rm lb, \it t}\left(e_k,\rho_k\right), \  \forall k\in\mathcal K.
\end{align}
Subsequently, a convex subset of non-convex constraint \eqref{P1-sub1-Eqv1_cons:b} is given by
\begin{align}\label{P1-sub1-Eqv1_cons:b_sca}
\tau_2 \leq \chi^{\rm lb, \it t}\left(e_k,\rho_k\right) - \frac{1}{2}\left(e_k^2 + \rho_k^2 \right), \ \forall k\in\mathcal K.
\end{align}
Second, to handle constraint \eqref{P1-sub1-Eqv1_cons:c}, we start by introducing slack variables $\{z_k\}$ and transforming \eqref{P1-sub1-Eqv1_cons:c} into 
\begin{subequations}\label{P1-sub1-Eqv1_cons:c_eqv}
\begin{align}
\zeta\sum_{i = 1}^K{\rm tr}\left(\mathbf a_{i,k,1}\mathbf a_{i,k,1}^H\mathbf W_{i,1}\right) + z_k^2 \geq E_k, \ \forall k \in\mathcal K, \label{P1-sub1-Eqv1_cons:c_1}\\
\frac{z_k^2}{1 - \rho_k} \leq \zeta\sum_{i = 1}^K{\rm tr}\left(\mathbf a_{i,k,2}\mathbf a_{i,k,2}^H\mathbf W_{i,2}\right), \ \forall k\in\mathcal K. \label{P1-sub1-Eqv1_cons:c_2}
\end{align}
\end{subequations}
It is not hard to see that if any constraint in \eqref{P1-sub1-Eqv1_cons:c_2} holds with strict inequality at the optimum, the corresponding $z_k^2$ can be increased to force this constraint hold with strict equality, yet without violating the constraints in \eqref{P1-sub1-Eqv1_cons:c_1} or decreasing the objective value. This verifies the equivalence between \eqref{P1-sub1-Eqv1_cons:c} and \eqref{P1-sub1-Eqv1_cons:c_eqv}. Notice that constraint \eqref{P1-sub1-Eqv1_cons:c_2} is convex while constraint \eqref{P1-sub1-Eqv1_cons:c_1} is not due to the convex term $z_k^2$. Similar to \eqref{P1-sub1-Eqv1_cons:b_sca}, we replace $z_k^2$ with its first-order Taylor expansion-based underestimator that yields the following convex subset of constraint \eqref{P1-sub1-Eqv1_cons:c_1}: \looseness=-1 
{
\begin{align}\label{P1-sub1-Eqv1_cons:c_1_sca}
\zeta\sum_{i = 1}^K{\rm tr}\left(\mathbf a_{i,k,1}\mathbf a_{i,k,1}^H\mathbf W_{i,1}\right) - \left(z_k^t\right)^2 + 2z_k^tz_k \geq E_k, \ \forall k \in\mathcal K,
\end{align}}%
where $z_k^t$ is the given local point in the $t$-th iteration. 

Combining the above results, a convex semidefinite program (SDP) for the non-convex problem \eqref{P1-sub1-Eqv1} in the $t$-th iteration is given by 
\begin{subequations}\label{P1-sub1-Eqv1-SCA}
	\begin{eqnarray}
	&\underset{\Delta}{\max}& \sum_{k=1}^K\left(f_{k,2} - g_{k,2}^{\rm ub, \it t}\left(\mathbf W_2, e_k, \tau_2 \right) + f_{k,3}\right. \nonumber\\
	&& \left. - g_{k,3}^{\rm ub, \it t}\left(\mathbf W_3, \tau_3 \right)\right) \\
	&\text{s.t.}& \eqref{P1-sub1-Eqv1_cons:b_sca}, \eqref{P1-sub1-Eqv1_cons:c_1_sca}, \eqref{P1-sub1-Eqv1_cons:c_2}, \eqref{P1-sub1-Eqv1_cons:d}, \eqref{P1_cons:d}, \eqref{P1_cons:e}, 
	\end{eqnarray}
\end{subequations}
where $\Delta \triangleq \left\lbrace \{\mathbf W_{i,j} \succeq \mathbf 0\}, \{\tau_j\}, \{\rho_k\}, \{e_k\}, \{z_k\}\right\rbrace$. Problem \eqref{P1-sub1-Eqv1-SCA} can be optimally solved by existing solvers such as CVX \cite{2004_S.Boyd_cvx}. It is worth noting that the solution of problem \eqref{P1-sub1-Eqv1-SCA} is also feasible for problem \eqref{P1-sub1-Eqv1}, but the reverse does not hold in general. Thus, we can obtain a lower bound on the optimal value of problem \eqref{P1-sub1-Eqv1} by solving problem \eqref{P1-sub1-Eqv1-SCA}. Besides, for problem \eqref{P1-sub1-Eqv1-SCA}, we have the following proposition. 
\begin{prop}\label{prop1}
	Suppose that problem \eqref{P1-sub1-Eqv1-SCA} is feasible for $P_i > 0$ and $E_k > 0$, $\forall i,k\in\mathcal K$, then its optimal solution, denoted by $\left\lbrace \{\mathbf W_{i,j}^\star\}, \{\tau_j^\star\}, \{\rho_k^\star\}, \{e_k^\star\}, \{z_k^\star\}\right\rbrace$, satisfies the following conditions:\\
	C1: ${\rm rank}\left(\mathbf W_{i,1}^\star\right) = 1$ if $\tau_1^{\star} > 0$ and $\{\mathbf a_{i,k,1}\}_{\forall k\in\mathcal K}$ are independently distributed, $\forall i\in\mathcal K$; \\
	C2: ${\rm rank}\left(\mathbf W_{i,2}^\star\right) = 1$ if $\tau_2^{\star} > 0$, ${\rm tr}\left(\mathbf W_{i,2}^{\star}\right) = \tau_2^{\star}P_i$, and $\{\mathbf a_{i,k,j}\}_{\forall k\in\mathcal K, j\in\{1,2\}}$ are independently distributed, $\forall i \in\mathcal K$;\\
	C3: ${\rm rank}\left(\mathbf W_{i,3}^\star\right) = 1$ if $\tau_3^{\star} > 0$, ${\rm tr}\left(\mathbf W_{i,3}^{\star}\right) = \tau_3^{\star}P_i$, and $\{\mathbf a_{i,k,3}\}_{\forall k\in\mathcal K}$ are independently distributed, $\forall i \in\mathcal K$. 
\end{prop}
\begin{proof}
	Please refer to the appendix.
\end{proof}	
\begin{rem}\label{rem1}
    It is worth noting that the above conditions of independently distributed $\{\mathbf a_{i,k,1}\}_{\forall k\in\mathcal K}$, $\{\mathbf a_{i,k,j}\}_{\forall k\in\mathcal K, j\in\{1,2\}}$, and $\{\mathbf a_{i,k,3}\}_{\forall k\in\mathcal K}$ are sufficient but unnecessary. For example, the independently distributed $\{\mathbf a_{i,k,1}\}_{\forall k\in\mathcal K}$ guarantees that the dominant eigenvalue of the matrix $\mathbf \Upsilon_{i,1}$ defined in the appendix has multiplicity $1$, but not vice versa. Actually, in almost all of our simulations, if $\tau_j^{\star} > 0$ for $j\in\{1,2,3\}$, then Tx $i$, $i\in\mathcal K$, employs either full or zero power and we almost always have ${\rm rank}\left(\mathbf W_{i,j}^\star\right) \leq 1$. How to theoretically characterize the optimal solution structure under looser conditions deserves further investigation. It is also worth mentioning that the reveal of the rank properties of $\{\mathbf W_{i,j}^\star\}$ in Proposition \ref{prop1} aims to offer important engineering insights, whereas the ranks of $\{\mathbf W_{i,j}^\star\}$ have no impact on the convergence of the proposed algorithm. 
\end{rem}

After obtaining the optimal solution, we can recover $\{\mathbf S_{i,j}^{\star}\}$ by setting $\mathbf S_{i,j}^{\star} = \frac{\mathbf W_{i,j}^{\star}}{\tau_j^{\star}}$ if $\tau_j^{\star} > 0$ and $\mathbf S_{i,j}^{\star} = \mathbf 0$ otherwise, $\forall i\in\mathcal K$, $j\in\{1,2,3\}$. 

\subsection{Optimizing $\left\lbrace \mathbf v_j \right\rbrace$ for Given $\left\lbrace \{\mathbf S_{i,j}\}, \{\tau_j\}, \{\rho_k\}\right\rbrace$} \label{Sec_P1_sub2}
For given $\left\lbrace \{\mathbf S_{i,j}\}, \{\tau_j\}, \{\rho_k\}\right\rbrace$, the subproblem of (P1) for only optimizing $\{\mathbf v_j\}$ can be expressed as 
\begin{subequations}\label{P1-sub2}
	\begin{align}
    & \hspace{-8mm}\underset{\{\mathbf v_j\}, j\in\{1,2,3\}}{\max} \hspace{2mm} \sum_{k=1}^K\Big[ \tau_2\log_2\left(1 + \bar\gamma_{k,2} \right) + \tau_3\log_2\left( 1 + \bar\gamma_{k,3}\right)\Big] \\
    \text{s.t.} \hspace{3mm}&\zeta\tau_1\sum_{i = 1}^K\mathbf v_1^H\mathbf A_{i,k,1}\mathbf v_1 + \zeta\tau_2\left(1 - \rho_k\right)\sum_{i = 1}^K\mathbf v_2^H\mathbf A_{i,k,2}\mathbf v_2 \nonumber\\
    & \geq E_k, \ \forall k \in\mathcal K, \label{P1_sub2_cons:b}\\
	& \eqref{P1_cons:f},
	\end{align}
\end{subequations} 
where 
\begin{align}
\bar\gamma_{k,2} & = \frac{\mathbf v_2^H\mathbf A_{k,k,2}\mathbf v_2}{\sum_{i=1,i\neq k}^K\mathbf v_2^H\mathbf A_{i,k,2}\mathbf v_2 + \tilde\sigma_k^2 + \frac{\hat\sigma_k^2}{\rho_k}}, \\
\bar\gamma_{k,3} & = \frac{\mathbf v_3^H\mathbf A_{k,k,3}\mathbf v_3}{\sum_{i=1,i\neq k}^K\mathbf v_3^H\mathbf A_{i,k,3}\mathbf v_3 + \sigma_k^2}, 
\end{align} 
with $\mathbf A_{i,k,j} = \mathbf H_{i,k}\mathbf S_{i,j}\mathbf H_{i,k}^H$, $i,k\in\mathcal K$, $j\in\{1,2,3\}$. Observe that problem \eqref{P1-sub2} is non-convex with a non-concave objective function and a non-convex feasible set. To facilitate the solution development of problem \eqref{P1-sub2}, we reformulate it by introducing slack variables $\{\mu_{k,2}\}$ and  $\{\mu_{k,3}\}$, $k\in\mathcal K$ as 
\begin{subequations}\label{P1-sub2-Eqv1}
	\begin{align}
    &\hspace{-7mm}\underset{\substack{\{\mathbf v_j\}, \{\mu_{k,2}\}, \{\mu_{k,3}\}, \\ j\in\{1,2,3\}, k\in\mathcal K}}{\max} \hspace{2mm} \sum_{k=1}^K\Big[\tau_2\log_2\left( 1 + \mu_{k,2}\right) \nonumber\\
    & \hspace{2cm} + \tau_3\log_2\left( 1 + \mu_{k,3}\right)\Big]  \\
    \text{s.t.} \hspace{3mm} & \frac{\mathbf v_2^H\mathbf A_{k,k,2}\mathbf v_2}{\mu_{k,2}} \geq \sum_{i=1,i\neq k}^K\mathbf v_2^H\mathbf A_{i,k,2}\mathbf v_2 + \tilde\sigma_k^2 + \frac{\hat\sigma_k^2}{\rho_k}, \nonumber\\
	& \forall k\in\mathcal K, \label{P1-sub2-Eqv1_cons:b}\\
	& \frac{\mathbf v_3^H\mathbf A_{k,k,3}\mathbf v_3}{\mu_{k,3}} \geq \sum_{i=1,i\neq k}^K\mathbf v_3^H\mathbf A_{i,k,3}\mathbf v_3 + \sigma_k^2, \ \forall k\in\mathcal K, \label{P1-sub2-Eqv1_cons:c}\\
	& \eqref{P1_sub2_cons:b}, \eqref{P1_cons:f}.
	\end{align}
\end{subequations}
The equivalence between problems \eqref{P1-sub2} and \eqref{P1-sub2-Eqv1} is guaranteed by the fact that all the constraints in \eqref{P1-sub2-Eqv1_cons:b} and \eqref{P1-sub2-Eqv1_cons:c} are active when the optimality of problem \eqref{P1-sub2-Eqv1} is attained. By such a reformulation, the objective function is now concave, but with new non-convex constraints \eqref{P1-sub2-Eqv1_cons:b} and \eqref{P1-sub2-Eqv1_cons:c} in addition to the original non-convex constraint \eqref{P1_sub2_cons:b}. Similarly, in the following, we tackle the non-convexity of these constraints based on the SCA technique as in the previous subsection. To this end, we define $\mathcal F_{\mathbf B}(\mathbf x, y) \triangleq \frac{\mathbf x^H\mathbf B\mathbf x}{y}$ and $\mathcal G_{\mathbf B}(\mathbf x) \triangleq \mathbf x^H\mathbf B\mathbf x$, where $\mathbf x\in\mathbb C^{M\times1}$, $y>0$ and $\mathbf B\succeq \mathbf 0$. It is not difficult to verify that $\mathcal F_{\mathbf B}(\mathbf x, y)$ and $\mathcal G_{\mathbf B}(\mathbf x)$ are both convex functions. Thus, based on the first-order Taylor expansion, their affine lower bounds can be established, denoted by $\mathcal F^{\rm lb, \it t}_{\mathbf B}(\mathbf x, y) \triangleq \frac{2{\rm Re}\left\lbrace \left( \mathbf x^t\right) ^H\mathbf B\mathbf x\right\rbrace }{y^t} - \frac{\left( \mathbf x^t\right) ^H\mathbf B\mathbf x^t}{\left( y^t\right)^2}y$ and $\mathcal G^{\rm lb, \it t}_{\mathbf B}(\mathbf x) \triangleq 2{\rm Re}\left\lbrace \left( \mathbf x^t\right) ^H\mathbf B\mathbf x\right\rbrace  - \left( \mathbf x^t\right) ^H\mathbf B\mathbf x^t$, respectively, with $\mathbf x^t\in\mathbb C^{M\times1} $ and $y^t > 0$ being any given points. Accordingly, given the local feasible points $\{\mathbf v_j^t\}$ and $\{\mu_{k,j}^t\}$ in the $t$-th iteration, we can approximate \eqref{P1-sub2-Eqv1_cons:b}, \eqref{P1-sub2-Eqv1_cons:c}, and \eqref{P1_sub2_cons:b} as 
\begin{align}
& \mathcal F^{\rm lb, \it t}_{\mathbf A_{k,k,2}}(\mathbf v_2, \mu_{k,2}) \geq \sum_{i=1,i\neq k}^K\mathbf v_2^H\mathbf A_{i,k,2}\mathbf v_2 + \tilde\sigma_k^2 + \frac{\hat\sigma_k^2}{\rho_k}, \nonumber\\
& \forall k\in\mathcal K, \label{P1-sub2-Eqv1_cons:b_sca}\\
& \mathcal F^{\rm lb, \it t}_{\mathbf A_{k,k,3}}(\mathbf v_3, \mu_{k,3}) \geq \sum_{i=1,i\neq k}^K\mathbf v_3^H\mathbf A_{i,k,3}\mathbf v_3 + \sigma_k^2, \ \forall k\in\mathcal K, \label{P1-sub2-Eqv1_cons:c_sca}\\
& \zeta\tau_1\sum_{i = 1}^K\mathcal G^{\rm lb, \it t}_{\mathbf A_{i,k,1}}(\mathbf v_1) + \zeta\tau_2\left(1 - \rho_k\right)\sum_{i = 1}^K\mathcal G^{\rm lb, \it t}_{\mathbf A_{i,k,2}}(\mathbf v_2)\geq E_k, \nonumber\\
& \forall k \in\mathcal K, \label{P1-sub2_cons:b_sca}
\end{align}
respectively, which are obtained by replacing all the convex terms in the left-hand-sides (LHSs) of \eqref{P1-sub2-Eqv1_cons:b}, \eqref{P1-sub2-Eqv1_cons:c}, and \eqref{P1_sub2_cons:b} with their respective lower bounds. 

Consequently, in the $t$-th iteration, a lower bound objective value of problem \eqref{P1-sub2-Eqv1} can be obtained by solving 
\begin{subequations}\label{P1-sub2-Eqv1-SCA}
	\begin{align}
    \underset{\substack{\{\mathbf v_j\}, \{\mu_{k,2}\}, \{\mu_{k,3}\}, \\ j\in\{1,2,3\}, k\in\mathcal K}}{\max} \hspace{2mm}& \sum_{k=1}^K\Big[\tau_2\log_2\left( 1 + \mu_{k,2}\right) \nonumber\\
	& + \tau_3\log_2\left( 1 + \mu_{k,3}\right)\Big]   \\
	\text{s.t.} \hspace{3mm}& \eqref{P1-sub2-Eqv1_cons:b_sca}-\eqref{P1-sub2_cons:b_sca}, \eqref{P1_cons:f}.
	\end{align}
\end{subequations}
Note that problem \eqref{P1-sub2-Eqv1-SCA} is a convex quadratically constrained quadratic program (QCQP) and its numerical solution can be found by off-the-shelf solvers, e.g., CVX \cite{2004_S.Boyd_cvx}.  

\subsection{Overall Algorithm} \label{Sec_P1_alg}
\begin{algorithm}[!tp]  
	\caption{AO algorithm for (P1)}  \label{Alg1}  
	\begin{algorithmic}[1]
		\STATE Set $t=0$. Initialize $\Delta^0$, $\{\mathbf v_j^0\}$, and $\{\mu_{k,j}^0\}$. 
		\REPEAT
		\STATE Obtain $\Delta^{t+1}$ by solving problem \eqref{P1-sub1-Eqv1-SCA} with given $\Delta^{t}$ and $\{\mathbf v_j^{t}\}$, and set $\mathbf S_{i,j}^{t+1} = {\mathbf W_{i,j}^{t+1}}/{\tau_j^{t+1}}$ if $\tau_j^{t+1} > 0$, $\forall i\in\mathcal K$, $j\in\{1,2,3\}$ and $\mathbf S_{i,j}^{t+1} = \mathbf 0$ otherwise. \label{Alg:solve_sub1}
		\STATE Obtain $\{\mathbf v_j^{t+1}\}$ and $\{\mu_{k,j}^{t+1}\}$ by solving problem \eqref{P1-sub2-Eqv1-SCA} with given $\left\lbrace \{\mathbf S_{i,j}^{t+1}\}, \{\tau_j^{t+1}\}, \{\rho_k^{t+1}\}\right\rbrace$, $\{\mathbf v_j^{t}\}$, and $\{\mu_{k,j}^t\}$. \label{Alg:solve_sub2}
		\STATE Set $t=t+1$.  \label{Alg:update_index}
		\UNTIL The fractional increase of the objective value of (P1) between two consecutive iterations is below a threshold $\epsilon > 0$. \label{Alg:stop}
	\end{algorithmic} 
\end{algorithm}

Using the above results, an overall algorithm for seeking a suboptimal solution to (P1) is provided in Algorithm \ref{Alg1}, where $\Delta^t \triangleq \left\lbrace \{\mathbf W_{i,j}^t\}, \{\tau_j^t\}, \{\rho_k^t\}, \{e_k^t\}, \{z_k^t\}\right\rbrace$. In the following, we provide a detailed convergence proof of Algorithm \ref{Alg1}. Firstly, define $R\left(\{\mathbf S_{i,j}\}, \{\tau_j\}, \{\rho_k\},\{\mathbf v_j\}\right)$, $r_1\left(\Delta\backslash\{z_k\}, \{\mathbf v_j\}\right)$, $r_1^{\rm lb}\left(\Delta, \{\mathbf v_j\}\right)$ as the objective values of problems (P1), \eqref{P1-sub1-Eqv1}, and \eqref{P1-sub1-Eqv1-SCA}, respectively. Then, in step \ref{Alg:solve_sub1} of Algorithm \ref{Alg1}, we have
\begin{align}\label{Alg1_proof_1}
& R\left(\{\mathbf S_{i,j}^t\}, \{\tau_j^t\}, \{\rho_k^t\},\{\mathbf v_j^t\}\right) \nonumber\\
& \overset{(a)}{=} r_1\left(\Delta^t\backslash\{z_k^t\}, \{\mathbf v_j^t\}\right) \nonumber\\
& \overset{(b)}{=} r_1^{\rm lb}\left(\Delta^t, \{\mathbf v_j^t\}\right) \nonumber\\
& \overset{(c)}{\leq} r_1^{\rm lb}\left(\Delta^{t+1}, \{\mathbf v_j^t\}\right) \nonumber\\ 
& \overset{(d)}{\leq} r_1\left(\Delta^{t+1}\backslash\{z_k^{t+1}\}, \{\mathbf v_j^t\}\right) \nonumber\\
& \overset{(e)}{=} R\left(\{\mathbf S_{i,j}^{t+1}\}, \{\tau_j^{t+1}\}, \{\rho_k^{t+1}\},\{\mathbf v_j^t\}\right), 
\end{align}
where $(a)$ and $(e)$ hold since problem \eqref{P1-sub1-Eqv1} is equivalent to problem \eqref{P1-sub1} (i.e., the first subproblem of (P1) under given $\{\mathbf v_j^t\}$); $(b)$ holds due to the fact that the first-order Taylor expansions in \eqref{Taylor_expansion_1} and \eqref{Taylor_expansion_2} are tight at given $\Delta^t$, respectively, which means that the objective value of problem \eqref{P1-sub1-Eqv1-SCA} at $\Delta^t$ is equal to that of problem \eqref{P1-sub1-Eqv1} at $\Delta^t\backslash\{z_k^t\}$; $(c)$ is true because $\Delta^{t+1}$ is the optimal solution of problem \eqref{P1-sub1-Eqv1-SCA} under given $\{\mathbf v_j^t\}$; $(d)$ holds since the optimal value of problem \eqref{P1-sub1-Eqv1-SCA} provides a lower bound for that of problem \eqref{P1-sub1-Eqv1}. Secondly, define $r_2\left(\{\mathbf S_{i,j}\}, \{\tau_j\}, \{\rho_k\},\{\mathbf v_j\},\{\mu_{k,j}\}\right)$ and $r_2^{\rm lb}\left(\{\mathbf S_{i,j}\}, \{\tau_j\}, \{\rho_k\},\{\mathbf v_j\},\{\mu_{k,j}\}\right)$ as the objective values of problems \eqref{P1-sub2-Eqv1} and \eqref{P1-sub2-Eqv1-SCA}, respectively. Then, in step \ref{Alg:solve_sub2} of Algorithm \ref{Alg1}, for reasons similar to those explained above for \eqref{Alg1_proof_1}, the following results hold: 
\begin{align}\label{Alg1_proof_2}
& R\left(\{\mathbf S_{i,j}^{t+1}\}, \{\tau_j^{t+1}\}, \{\rho_k^{t+1}\},\{\mathbf v_j^t\}\right) \nonumber\\
& = r_2\left(\{\mathbf S_{i,j}^{t+1}\}, \{\tau_j^{t+1}\}, \{\rho_k^{t+1}\},\{\mathbf v_j^t\},\{\mu_{k,j}^t\}\right) \nonumber\\ 
& = r_2^{\rm lb}\left(\{\mathbf S_{i,j}^{t+1}\}, \{\tau_j^{t+1}\}, \{\rho_k^{t+1}\},\{\mathbf v_j^t\},\{\mu_{k,j}^t\}\right) \nonumber\\
& \leq r_2^{\rm lb}\left(\{\mathbf S_{i,j}^{t+1}\}, \{\tau_j^{t+1}\}, \{\rho_k^{t+1}\},\{\mathbf v_j^{t+1}\},\{\mu_{k,j}^{t+1}\}\right)  \nonumber\\
& \leq r_2\left(\{\mathbf S_{i,j}^{t+1}\}, \{\tau_j^{t+1}\}, \{\rho_k^{t+1}\},\{\mathbf v_j^{t+1}\},\{\mu_{k,j}^{t+1}\}\right) \nonumber\\
& = R\left(\{\mathbf S_{i,j}^{t+1}\}, \{\tau_j^{t+1}\}, \{\rho_k^{t+1}\},\{\mathbf v_j^{t+1}\}\right). 
\end{align}
Based on \eqref{Alg1_proof_1} and \eqref{Alg1_proof_2}, we have
\begin{align}
& R\left(\{\mathbf S_{i,j}^t\}, \{\tau_j^t\}, \{\rho_k^t\},\{\mathbf v_j^t\}\right) \nonumber\\
& \leq R\left(\{\mathbf S_{i,j}^{t+1}\}, \{\tau_j^{t+1}\}, \{\rho_k^{t+1}\},\{\mathbf v_j^{t+1}\}\right),
\end{align}
which indicates that by repeating steps \ref{Alg:solve_sub1}-\ref{Alg:update_index} of Algorithm \ref{Alg1}, we can obtain a non-decreasing sequence of objective values of (P1). Additionally, the optimal value of (P1) is bounded from above. Thus, Algorithm \ref{Alg1} is guaranteed to converge.

In each iteration, the main computational complexity of Algorithm 1 lies in solving the SDP problem \eqref{P1-sub1-Eqv1-SCA} and the QCQP problem \eqref{P1-sub2-Eqv1-SCA}. According to the complexity analyses in \cite{2010_Imre_SDR_complexity} and \cite{2014_K.wang_complexity}, the computational cost of solving problem \eqref{P1-sub1-Eqv1-SCA} is $\mathcal O\left(\sqrt{M}\ln\left(\frac{1}{\varepsilon}\right) \left(KM^3+K^2M^2+K^3\right)\right)$ and that of solving problem \eqref{P1-sub2-Eqv1-SCA} is $\mathcal O\left(\sqrt{KN}\ln\left(\frac{1}{\varepsilon}\right)\left(KN^4 + K^2N^3 + K^3N^2\right) \right)$, with $\varepsilon > 0$ denoting the given solution accuracy. Therefore, the computational complexity of each iteration of Algorithm 1 is about $\mathcal O\Big[\ln\left(\frac{1}{\varepsilon}\right)\Big(  \sqrt{M}\left(KM^3+K^2M^2+K^3\right) +  \sqrt{KN}(KN^4 + K^2N^3 + K^3N^2)\Big)\Big]$.

\section{Proposed Algorithms for (P2) and (P2-D)}\label{Sec_P2_solution}
\subsection{Proposed Algorithm for (P2)}
Similar to (P1), we leverage the AO method to suboptimally solve (P2) by decomposing it into two subproblems, as detailed below. 

\subsubsection{Optimizing $\left\lbrace\{\mathbf S_{i,j}\}, \{\tau_j\} \right\rbrace$ for Given $\{\mathbf v_j\}$} \label{Sec_P2_sub1}
Given any feasible $\{\mathbf v_j\}$, $\left\lbrace\{\mathbf S_{i,j}\}, \{\tau_j\} \right\rbrace$ can be optimized by solving (P2) with the constraints in \eqref{P2_cons:b}-\eqref{P2_cons:d}. By comparing the objective functions of this subproblem and the first subproblem of (P1), we find that the former is similar to but much simpler than the latter. In view of this, instead of applying the change of variables at the very beginning as in Section \ref{Sec_P1_sub1}, we start by rewriting the objective function $\sum_{k=1}^K\tau_k\log_2\left( 1 + \gamma_k\right)$ in an equivalent form as $\sum_{k=1}^K\tau_k\left(o_k - q_k\right)$, where $o_k$ and $q_k$ are both concave functions given by $o_k = \log_2\left(\sum_{i=1}^K{\rm tr}\left(\mathbf b_{i,k}\mathbf b_{i,k}^H \mathbf S_{i,k}\right) + \sigma_k^2\right)$, and $q_k = \log_2\left(\sum_{i=1,i\neq k}^K{\rm tr}\left(\mathbf b_{i,k}\mathbf b_{i,k}^H \mathbf S_{i,k}\right) + \sigma_k^2\right)$, respectively, $\forall k\in\mathcal K$. In the $r$-th iteration, by applying the first-oder Taylor expansion at the given points $\{\mathbf S_{i,k}^r\}$, $i,k\in\mathcal K$ to $q_k$, we have
\begin{align}
q_k\left(\mathbf S_k\right) & \leq \log_2\left(\Lambda_k^r\right)  + \frac{\sum\nolimits_{i=1,i\neq k}^K{\rm tr}\left(\mathbf b_{i,k}\mathbf b_{i,k}^H\left(\mathbf S_{i,k} - \mathbf S_{i,k}^r\right) \right)}{\Lambda_k^r\ln2} \nonumber\\
& \triangleq q_k^{\rm ub, \it r}\left(\mathbf S_k\right), \ \forall k\in\mathcal K, 
\end{align}
where $\mathbf S_k = \{\mathbf S_{i,k}\}_{\forall i\in\mathcal K\backslash\{k\}}$ and $\Lambda_k^r = \sum_{i=1,i\neq k}^K{\rm tr}\left(\mathbf b_{i,k}\mathbf b_{i,k}^H \mathbf S_{i,k}^r\right) + \sigma_k^2$. Then, we replace the objective function $\sum_{k=1}^K\tau_k\log_2\left( 1 + \gamma_k\right)$ with its lower bound $\sum_{k=1}^K\tau_k\left(o_k - q_k^{\rm ub, \it r}\left(\mathbf S_k\right)\right)$, yielding 
\begin{subequations}\label{P2-sub1-SCA}
\begin{align}
	\underset{\substack{\{\mathbf S_{i,j} \succeq \mathbf 0\}, \{\tau_j\},\\ i,j\in\mathcal K}}{\max} \hspace{2mm}& \sum_{k=1}^K\tau_k\left(o_k - q_k^{\rm ub, \it r}\left(\mathbf S_k\right)\right)  \\
	\text{s.t.} \hspace{3mm}& \eqref{P2_cons:b} - \eqref{P2_cons:d}.
\end{align}
\end{subequations}
The key observation is that although problem \eqref{P2-sub1-SCA} is non-convex in its current form due to the non-concave objective function and non-convex constraint \eqref{P2_cons:b}, it can be equivalently transformed into the following convex SDP by simply applying the change of variables $\mathbf W_{i,j} = \tau_j\mathbf S_{i,j}$, $\forall i,j\in\mathcal K$. 
\begin{subequations}\label{P2-sub1-SCA-Eqv}
	\begin{align}
	\underset{\substack{\{\mathbf W_{i,j} \succeq \mathbf 0\},\\\{\tau_j\}, i,j\in\mathcal K}}{\max} & \hspace{1mm}\sum_{k=1}^K\left[ \tau_k\log_2\left(\frac{\sum_{i=1}^K{\rm tr}\left(\mathbf b_{i,k}\mathbf b_{i,k}^H \mathbf W_{i,k}\right)}{\tau_k} + \sigma_k^2\right)\right.  \nonumber\\
	& - \tau_k\log_2\left(\Lambda_k^r\right) \nonumber\\
	& \left. - \frac{\sum\nolimits_{i=1,i\neq k}^K{\rm tr}\left(\mathbf b_{i,k}\mathbf b_{i,k}^H\left(\mathbf W_{i,k} - \tau_k\mathbf S_{i,k}^r\right) \right)}{\Lambda_k^r\ln2}\right]  \\
	\text{s.t.}\hspace{3mm}& \sum_{j=1,j\neq k}^K\zeta\sum_{i=1}^K{\rm tr}\left(\mathbf c_{i,k,j}\mathbf c_{i,k,j}^H\mathbf W_{i,j}\right) \geq E_k, \nonumber\\
	& \forall k \in\mathcal K, \label{P2-sub1-SCA-Eqv_cons:b}\\
	& {\rm tr}\left(\mathbf W_{i,j} \right) \leq \tau_jP_i, \ \forall i,j\in\mathcal K, \\
	& \eqref{P2_cons:d}.
	\end{align}
\end{subequations}
The optimal solution of problem \eqref{P2-sub1-SCA-Eqv}, denoted by $\left\lbrace \{\mathbf W_{i,j}^{\star}\}, \{\tau_j^{\star}\}\right\rbrace$, can be found by standard solvers such as CVX \cite{2004_S.Boyd_cvx}. Furthermore, we have the following proposition.
\begin{prop}\label{prop2}
	Assuming that problem \eqref{P2-sub1-SCA-Eqv} is feasible for $P_i > 0$ and $E_k > 0$, $\forall i,k\in\mathcal K$, then if $\tau_k^{\star} > 0$, $k\in\mathcal K$, it follows that:\\
	C1: ${\rm rank}\left(\mathbf W_{k,k}^\star\right) = 1$ if $\{\mathbf b_{k,k}, \mathbf c_{k,k',k}\}_{\forall k'\in\mathcal K\backslash\{k\}}$ are independently distributed;\\
	C2: ${\rm rank}\left(\mathbf W_{i,k}^\star\right) = 1$ if ${\rm tr}\left(\mathbf W_{i,k}^{\star}\right) = \tau_k^{\star}P_i$ and $\{\mathbf b_{i,k}, \mathbf c_{i,k',k}\}_{\forall k'\in\mathcal K\backslash\{k\}}$ are independently distributed, $\forall i\in\mathcal K$, $i\neq k$. 
\end{prop}
\begin{proof}
	The detailed proof is similar to that in the appendix for Proposition \ref{prop1} and is omitted here for brevity. 
\end{proof}
Once $\left\lbrace \{\mathbf W_{i,j}^{\star}\}, \{\tau_j^{\star}\}\right\rbrace$ is obtained by solving problem \eqref{P2-sub1-SCA-Eqv}, we set $\mathbf S_{i,j}^{\star} = \frac{\mathbf W_{i,j}^{\star}}{\tau_j^{\star}}$ if $\tau_j^{\star} > 0$ and $\mathbf S^{\star}_{i,j} = \mathbf 0$ otherwise, $\forall i,j\in\mathcal K$.

\subsubsection{Optimizing $\{\mathbf v_j\}$ for Given $\{\mathbf S_{i,j}\}$ and $\{\tau_j\}$} By introducing slack variables $\{\mu_k\}$, the subproblem of (P2) w.r.t. $\{\mathbf v_j\}$ can be equivalently expressed as
\begin{subequations}\label{P2-sub2}
	\begin{align}
	&\hspace{-8mm}\underset{\{\mathbf v_j\}, \{\mu_k\}, j,k\in\mathcal K}{\max}\hspace{2mm} \sum_{k=1}^K\tau_k\log_2\left(1 + \mu_k\right)\\
	\text{s.t.} \hspace{3mm} & \frac{\mathbf v_k^H\mathbf A_{k,k,k}\mathbf v_k}{\mu_k} \geq \sum_{i=1,i\neq k}^K\mathbf v_k^H\mathbf A_{i,k,k}\mathbf v_k + \sigma_k^2,\ \forall k\in\mathcal K, \label{P2-sub2_cons:b}\\
	& \sum_{j=1,j\neq k}^K\zeta\tau_j\sum_{i=1}^K\mathbf v_j^H\mathbf A_{i,k,j}\mathbf v_j\geq E_k, \ \forall k \in\mathcal K, \label{P2-sub2_cons:c}\\
	& \eqref{P2_cons:e},
	\end{align}
\end{subequations} 
where $\mathbf A_{i,k,j} = \mathbf H_{i,k}\mathbf S_{i,j}\mathbf H_{i,k}^H$, $i,k,j\in\mathcal K$ and the newly introduced constraint \eqref{P2-sub2_cons:b} must be satisfied with equality at the optimum. The non-convexity of problem \eqref{P2-sub2} stems from non-convex constraints \eqref{P2-sub2_cons:b}  and \eqref{P2-sub2_cons:c}. Since problem \eqref{P2-sub2} has a similar form as problem \eqref{P1-sub2-Eqv1}, it can be handled in the same manner as for \eqref{P1-sub2-Eqv1}. Specifically, following similar steps to \eqref{P1-sub2-Eqv1_cons:b_sca}-\eqref{P1-sub2_cons:b_sca}, in the $r$-th iteration of the proposed algorithm, problem \eqref{P2-sub2} can be approximated as 
\begin{subequations}\label{P2-sub2-SCA}
	\begin{align}
    & \hspace{-5mm}\underset{\{\mathbf v_j\},\{\mu_k\}, j,k\in\mathcal K}{\max} \hspace{2mm} \sum_{k=1}^K\tau_k\log_2\left(1 + \mu_k\right)\\
	\text{s.t.} \hspace{3mm}&  \mathcal F^{\rm lb, \it r}_{\mathbf A_{k,k,k}}(\mathbf v_k, \mu_k) \geq \sum_{i=1,i\neq k}^K\mathbf v_k^H\mathbf A_{i,k,k}\mathbf v_k + \sigma_k^2,\ \forall k\in\mathcal K, \label{P2-sub2-SCA_cons:b}\\
	& \sum_{j=1,j\neq k}^K\zeta\tau_j\sum_{i=1}^K\mathcal G^{\rm lb, \it r}_{\mathbf A_{i,k,j}}(\mathbf v_j)\geq E_k, \ \forall k \in\mathcal K, \label{P2-sub2-SCA_cons:c}\\
	& \eqref{P2_cons:e},
	\end{align}
\end{subequations} 
\normalsize where the expressions of $\mathcal F^{\rm lb, \it r}_{\mathbf A_{k,k,k}}(\mathbf v_k, \mu_k)$ and $\mathcal G^{\rm lb, \it r}_{\mathbf A_{i,k,j}}(\mathbf v_j)$ can be found in Section \ref{Sec_P1_sub2} after \eqref{P1-sub2-Eqv1} with index ``$t$'' replaced by ``$r$''. By direct inspection, problem \eqref{P2-sub2-SCA} is a convex QCQP that can be solved exactly using readily available solvers such as CVX \cite{2004_S.Boyd_cvx}. 

\subsubsection{Overall Algorithm} The detailed steps of the proposed algorithm for (P2) are omitted here due to the similarity between them and those listed in Algorithm \ref{Alg1}. 
In addition, this algorithm is guaranteed to converge for the same reasons as Algorithm \ref{Alg1}. 
Moreover, similar to the analysis in Section \ref{Sec_P1_alg}, the computational complexity of each iteration of this algorithm is in the order of $\mathcal O\left( \ln\left(\frac{1}{\varepsilon}\right) \left( \sqrt{M}\left(K^2M^3+K^4M^2+K^6\right) + K^6N^{\frac{9}{2}}\right)\right) $ \cite{2010_Imre_SDR_complexity,2014_K.wang_complexity}.

\subsection{Proposed Algorithm for (P2-D)} 
Since the only difference between (P2) and (P2-D) is that the objective function of (P2-D) is simpler, the above AO algorithm proposed for (P2) can be applied to solve (P2-D) but with slight modifications. To be specific, for (P2-D), we can exploit the hidden convexity of the subproblem w.r.t. $\left\lbrace\{\mathbf S_{i,j}\}, \{\tau_j\} \right\rbrace$ by only applying the change of variables (i.e., defining $\mathbf W_{i,j} = \tau_j\mathbf S_{i,j}$, $\forall i,j\in\mathcal K$). On the other hand, similar to \eqref{P2-sub2-SCA}, the subproblem of (P2-D) w.r.t. $\{\mathbf v_j\}$ can be approximated as a convex QCQP, which differs from \eqref{P2-sub2-SCA} in that the first constraint is $\mathcal G^{\rm lb, \it r}_{\mathbf A_{k,k,k}}(\mathbf v_k) \geq \mu_k\sigma_k^2,\ \forall k\in\mathcal K$.  

\section{Feasibility Checking and Initialization Methods}\label{Sec_feasibility}
Before executing the proposed algorithms, we have to check whether these schemes are feasible under the conflicting EH and transmit power constraints. If the answer is yes, how do we construct the initial points for these algorithms? To this end, for the IRS-aided hybrid TS-PS scheme, we consider the following optimization problem: 
\begin{subequations}\label{P1_feasi}
	\begin{align} 
	& \hspace{-9mm}\underset{\{\mathbf S_{i,1} \succeq \mathbf 0\}_{i\in\mathcal K}, \mathbf v_1, \delta}{\max}  \hspace{2mm} \delta \\
	\text{s.t.} \hspace{3mm} & \zeta\sum_{i = 1}^K{\rm tr}\left(\mathbf a_{i,k,1}\mathbf a_{i,k,1}^H\mathbf S_{i,1}\right) \geq \delta E_k, \ \forall k \in\mathcal K, \label{P1_feasi_cons:b}\\
	& {\rm tr}\left(\mathbf S_{i,1} \right) \leq P_i, \ \forall i\in\mathcal K, \\
	& \left|\left[\mathbf v_1\right]_n \right| \leq 1,  \ \left[\mathbf v_1\right]_{N+1} = 1, \ \forall n\in\mathcal N, \label{P1_feasi_cons:d}
	\end{align}
\end{subequations} 
where $\delta \triangleq \frac{1}{\tau_1}$. It is not hard to see that problem \eqref{P1_feasi} aims to use a minimal time fraction to satisfy the EH requirements at the Rxs. Since this problem is non-convex with $\{\mathbf S_{i,1}\}$ and $\mathbf v_1$ coupled in constraint \eqref{P1_feasi_cons:b}, we solve it suboptimally by alternately optimizing $\{\{\mathbf S_{i,1}\}, \delta\}$ and $\mathbf v_1$. The subproblem w.r.t. $\{\{\mathbf S_{i,1}\}, \delta\}$ under randomly generated $\mathbf v_1$ is a convex SDP,  while the subproblem w.r.t. $\mathbf v_1$ is non-convex but can be approximated as the following convex QCQP by applying the SCA technique:   
\begin{subequations}\label{P1_feasi_sub2_sca}
	\begin{align} \underset{\mathbf v_1, \delta'\geq 0}{\max}\hspace{2mm}& \delta + \delta' \\
	\text{s.t.}\hspace{3mm}& \zeta\sum_{i = 1}^K\mathcal G^{\rm lb, \it t}_{\mathbf A_{i,k,1}}(\mathbf v_1) \geq (\delta + \delta')E_k, \ \forall k \in\mathcal K, \\
	& \eqref{P1_feasi_cons:d}, 
	\end{align}
\end{subequations} 
where $\delta'$ is a newly introduced ``residual'' variable that makes the problem more efficient than the original feasibility-check subproblem in terms of the converged solution (please refer to \cite{2019_Qingqing_Joint} for a detailed explanation). By alternately optimizing $\{\{\mathbf S_{i,1}\}, \delta\}$ and $\mathbf v_1$, we can obtain a non-decreasing sequence of objective values. Once  the reciprocal of the objective value is less than $1$, we verify that the IRS-aided TS-PS scheme is feasible, denote the corresponding solution as $\Omega \triangleq \{\{\tilde{\mathbf S}_{i,1}\},\tilde{\delta}, \tilde{\mathbf v}_1\}$, and stop the iterations. Recall that we have to initialize $\Delta^0 \triangleq \left\lbrace \{\mathbf W_{i,j}^0\}, \{\tau_j^0\}, \{\rho_k^0\}, \{e_k^0\}, \{z_k^0\}\right\rbrace$ and $\{\mathbf v_j^0\}$, $\forall j\in\{1,2,3\}$, when executing Algorithm \ref{Alg1} for the IRS-aided hybrid TS-PS scheme. After obtaining $\Omega$, we set $\tau_1^0 = 1/\tilde{\delta}$, $\tau_j^0 = \left(1 - \tau_1^0\right)/2$, $\mathbf W_{i,1}^0 = \tau_1^0\tilde{\mathbf S}_{i,1}$, $\mathbf W_{i,j}^0 = \tau_j^0\sqrt{P_i}\mathbf a_{k,k,j}\mathbf a_{k,k,j}^H/\left\| \mathbf a_{k,k,j}\right\|^2$, $\rho_k^0 = 1$, $e_k^0 = \tau_2^0$, $z_k^0 = 0$, $\mathbf v_1^0 = \tilde{\mathbf v}_1$, and randomly generate $\mathbf v_j^0$, $\forall i,k\in\mathcal K$, $j\in\{2,3\}$. It is easy to see that this constructed initial point is feasible for problem \eqref{P1-sub1-Eqv1-SCA} (and also (P1)). For another case where the reciprocal of the objective value of problem \eqref{P1_feasi} is larger than $1$ when the proposed algorithm for \eqref{P1_feasi} converges, we consider that this scheme is infeasible under this setup.  

Similarly, to check the feasibility of the two IRS-aided TDMA-based schemes, we consider the following time minimization problem:
\begin{subequations}
	\begin{align}
    &\hspace{-1cm}\underset{\substack{\{\mathbf S_{i,j} \succeq \mathbf 0\}, \{\tau_j\},\\ \{\mathbf v_j\}, i,j\in\mathcal K}}{\min} \hspace{2mm} \sum_{j=1}^K\tau_j \\
	\text{s.t.} \hspace{3mm}&  \sum_{j=1,j\neq k}^K\zeta\tau_j\sum_{i=1}^K{\rm tr}\left(\mathbf c_{i,k,j}\mathbf c_{i,k,j}^H\mathbf S_{i,j}\right) \geq E_k, \ \forall k \in\mathcal K,\\
	& {\rm tr}\left(\mathbf S_{i,j} \right) \leq P_i, \ \forall i,j\in\mathcal K, \\
	& \left|\left[\mathbf v_j\right]_n \right| \leq 1, \ \left[\mathbf v_j\right]_{N+1} = 1, \ \forall n\in\mathcal N, j \in \mathcal K. 
	\end{align}
\end{subequations}
Note that the proposed algorithm for (P2) can be easily modified to solve this non-convex problem. The details are omitted here for brevity. If the obtained objective value is not larger than $1$, then the two TDMA-based schemes are feasible, and the obtained solution can be directly used as feasible initial points for the proposed algorithms for (P2) and (P2-D). Otherwise, these two schemes are considered infeasible.

\section{Simulation Results}\label{Sec_simulation}
\begin{figure}[!t]
	\centering
	\includegraphics[width=0.45\textwidth]{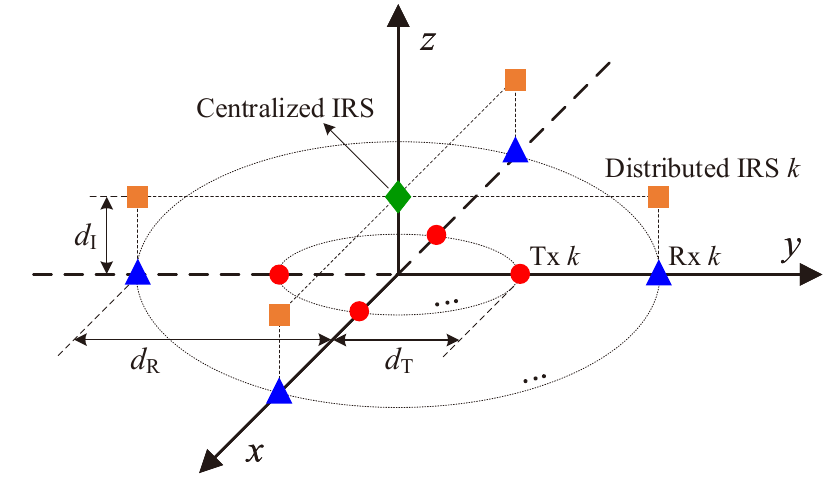}
	\caption{Simulation setup. The Txs, Rxs, distributed IRSs, and centralized IRS are marked by red `$\bullet$'s, blue `$\blacktriangle$'s, orange `$\blacksquare$'s, and green `$\blacklozenge$', respectively. }
	\label{fig:simulation_setup}
\end{figure}

In this section, we numerically demonstrate the efficacy of our proposed schemes. As depicted in Fig. \ref{fig:simulation_setup}, a three-dimensional (3D) coordinate setup is considered, where the $k$-th Tx and Rx are placed in spherical polar coordinates $\left(d_{\rm T}, \frac{2\pi(k-1)}{K}, \frac{\pi}{2} \right)$ and $\left(d_{\rm R}, \frac{2\pi(k-1)}{K}, \frac{\pi}{2} \right)$ in meters (m), respectively. In our simulations, we consider two practical IRS deployment strategies, the distributed and centralized, both of which have $N$ total IRS elements. For the distributed deployment, we assume that the number of IRSs is equal to that of the Tx-Rx pairs and each IRS is equipped with the same number of reflecting elements. The $k$-th distributed IRS is $d_{\rm I}$ m directly above the $k$-th Rx while the centralized IRS is located at $\left(d_{\rm I}, \frac{\pi}{2}, 0\right)$ in m. We adopt the distance-dependent large-scale path loss model in \cite{2019_Qingqing_Joint}, where the path loss at the reference distance of $1$ m is set to be $-30$ dB and the path loss exponents are set equal to $2.2$ for all the IRS-related links and to $3.5$ for all the other links. Moreover, it is assumed that all the IRS-related links undergo Rician fading with a Rician factor of 3 dB while all the other links experience Rayleigh fading. We assume that all the Rxs have identical EH requirements, i.e., $E_k = E$, $\forall k\in\mathcal K$. Unless otherwise stated, simulations are performed with the following parameters: $P_i = 23$ dBm, $\hat{\sigma}_k = \tilde{\sigma}_k^2 = \frac{1}{2}\sigma_k^2 = \frac{10^{-8}}{2}$ watt (W) \cite{2014_Chao_IFC}, $\forall i,k\in\mathcal K$, $\zeta = 70\%$ \cite{2014_Valenta_EH_efficiency}, $\epsilon = 10^{-4}$, $d_{\rm R} = 6$ m, and $d_{\rm I} = 1$ m. Besides, for comparison purposes, we consider two other transmission strategies, PS and TS, which can be obtained from our proposed hybrid TS-PS scheme by setting $\tau_1 = \tau_3 = 0$ and $\tau_2 = 0$, respectively. Note that for the case without IRSs, the TS scheme is the same as the TDMS scheme in \cite{2014_Chao_IFC}. Whenever a scheme is infeasible, we set its achievable sum rate to zero to account for the associated penalty. Each point on the simulation curves is obtained by averaging $500$ independent channel realizations. 

\begin{figure}[!t]
	\centering
	\subfigure[$d_{\rm T} = 0$ m.]{\label{fig:rate_vs_N_dTx0}
		\includegraphics[width=0.46\textwidth]{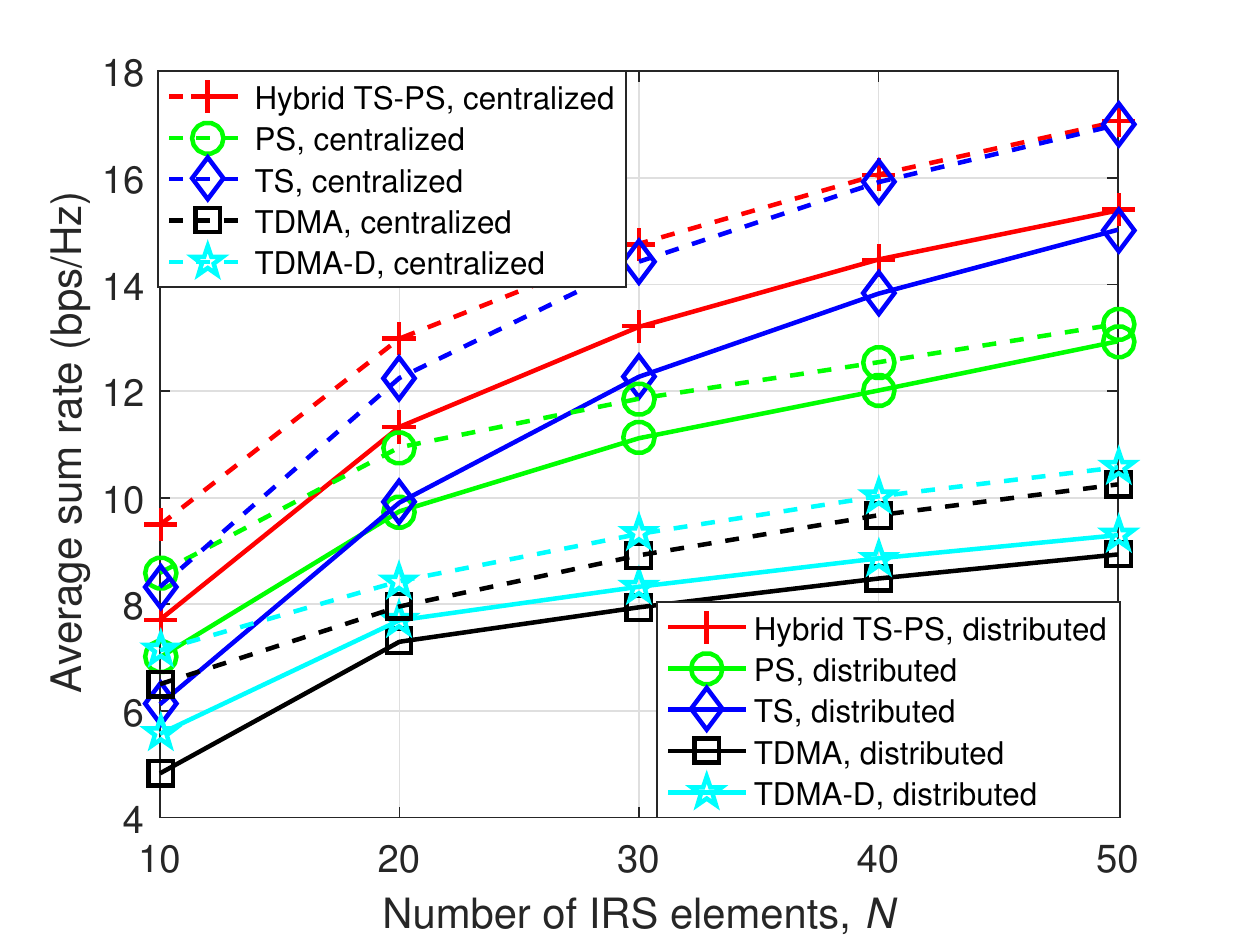}}
	\subfigure[$d_{\rm T} = 2$ m.]{\label{fig:rate_vs_N_dTx2}
		\includegraphics[width=0.46\textwidth]{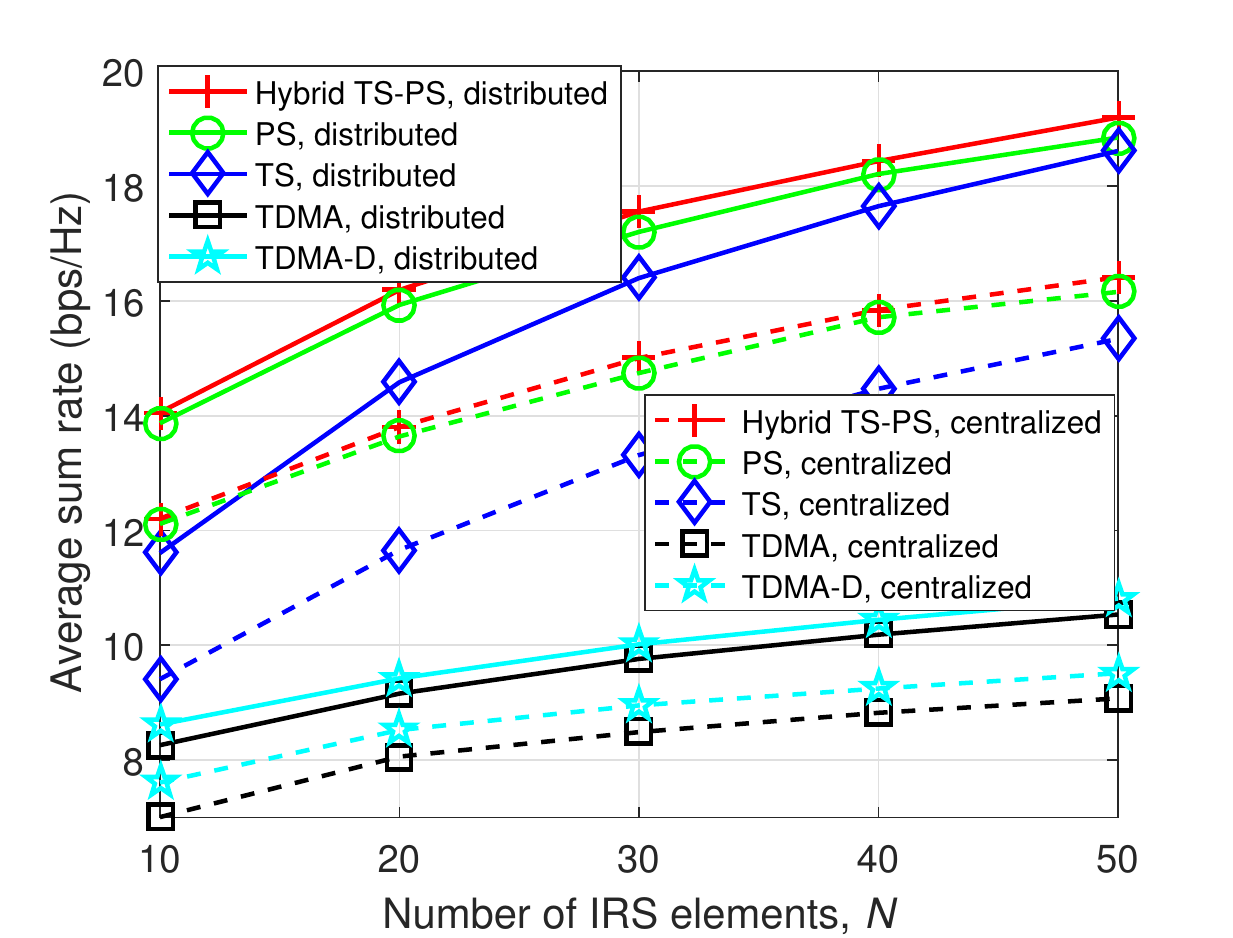}}
	\caption{Average sum rate versus the number of IRS elements for $M = K = 2$ and $E = 0.5$ $\mu\rm W$.}
	\label{fig:rate_vs_N}
\end{figure}

\subsection{Impact of Number of IRS Elements and IRS Deployment}
We first consider a setup where $M = K = 2$ and $E = 0.5$ microwatt ($\mu\rm W$). Fig. \ref{fig:rate_vs_N} plots the average sum rates of different schemes versus the number of IRS elements when $d_{\rm T} = 0$ and $2$ m, respectively. Both the distributed and centralized IRS deployment strategies are considered. Firstly, it is observed that the sum rates achieved by all the schemes increase with $N$, since more DoF at the IRS(s) can be exploited to establish a more favorable propagation environment. Nevertheless, there are diminishing returns in the sum rate gains achieved by the two TDMA-based schemes when $N > 30$. This is because the performance bottleneck would become the short information transmission duration of each Tx, which can hardly be extended even if $N$ is increased. Secondly, as expected, the hybrid TS-PS scheme always performs better than or as well as its sub-schemes, PS and TS. Also, the TDMA-D scheme consistently outperforms the TDMA scheme because the former is free from the cross-link interference compared to the latter. Lastly, we note that when $d_{\rm T} = 0$ m, the centralized deployment outperforms the distributed deployment, while the opposite is true for the case of $d_{\rm T} = 2$ m. In other words, neither a particular deployment strategy can always dominate the other in terms of achievable sum rate. To acquire further insights, we plot in Fig. \ref{fig:rate_vs_dTx_centr_distr} the average sum rate versus $d_{\rm T}$ when $N = 30$. From Fig. \ref{fig:rate_vs_dTx_centr_distr}, it can be seen that the performance gain of the centralized deployment over the distributed deployment decreases as $d_{\rm T}$ increases, and eventually, the performance of the distributed deployment surpasses that of the centralized deployment. This is expected since the distance between the centralized IRS and each Tx increases with $d_{\rm T}$, while each distributed IRS is always close to its corresponding Rx. It is worth noting that the sum rates of all the schemes adopting the distributed IRS deployment increase with $d_{\rm T}$ due to the increased channel power of the direct and reflected links. By contrast, the sum rates of all the schemes adopting the centralized IRS deployment initially decrease and then increase as $d_{\rm T}$ varies from $0$ m to $3$ m. This is because when $d_{\rm T}$ is small, the performance of these schemes is dominated by the reflected links whose channel power decreases with $d_{\rm T}$; when $d_{\rm T}$ is large, the performance of these schemes is dominated by the direct links whose channel power increases with $d_{\rm T}$.  

\begin{figure}[!t]
	\centering
	\includegraphics[width=0.46\textwidth]{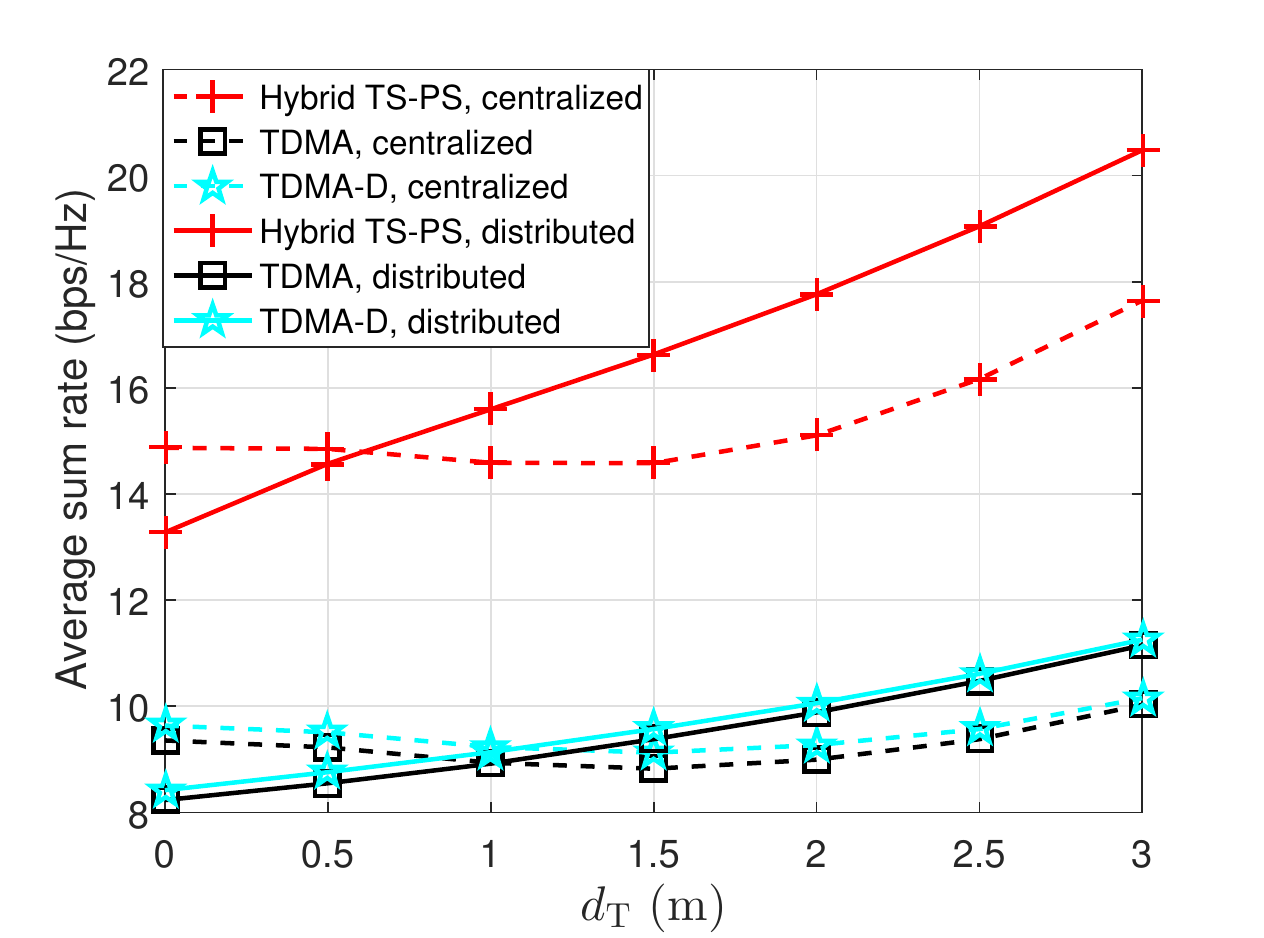}
	\caption{Average sum rate versus $d_{\rm T}$ for $M = K = 2$, $E = 0.5$ $\mu$W, and $N = 30$.}
	\label{fig:rate_vs_dTx_centr_distr}
\end{figure}

\begin{figure}[!t]
	\subfigure[]{\label{fig:time_fraction_vs_N}
		\includegraphics[width=0.46\textwidth]{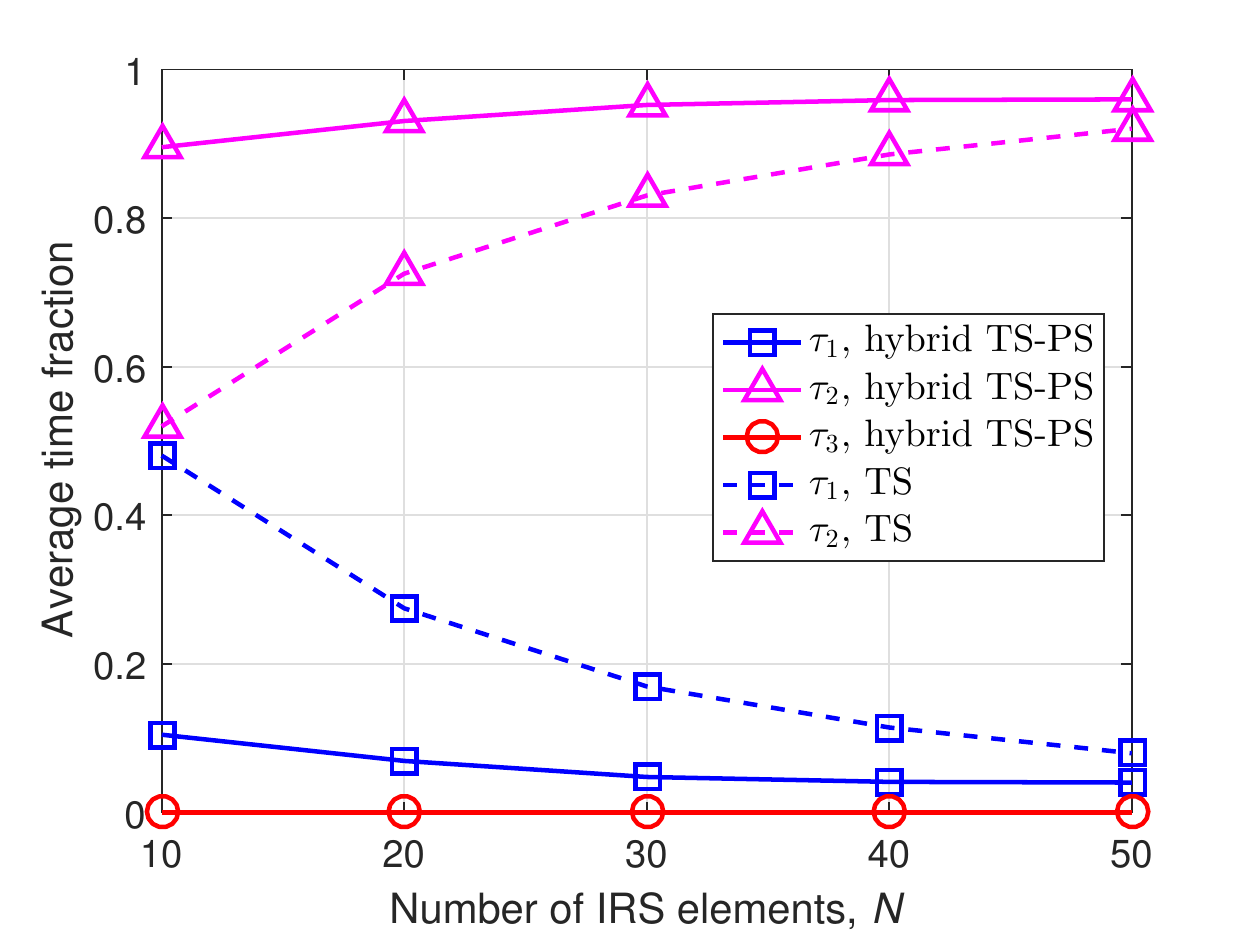}}
	\subfigure[]{\label{fig:PS_ratio_vs_N}
		\includegraphics[width=0.46\textwidth]{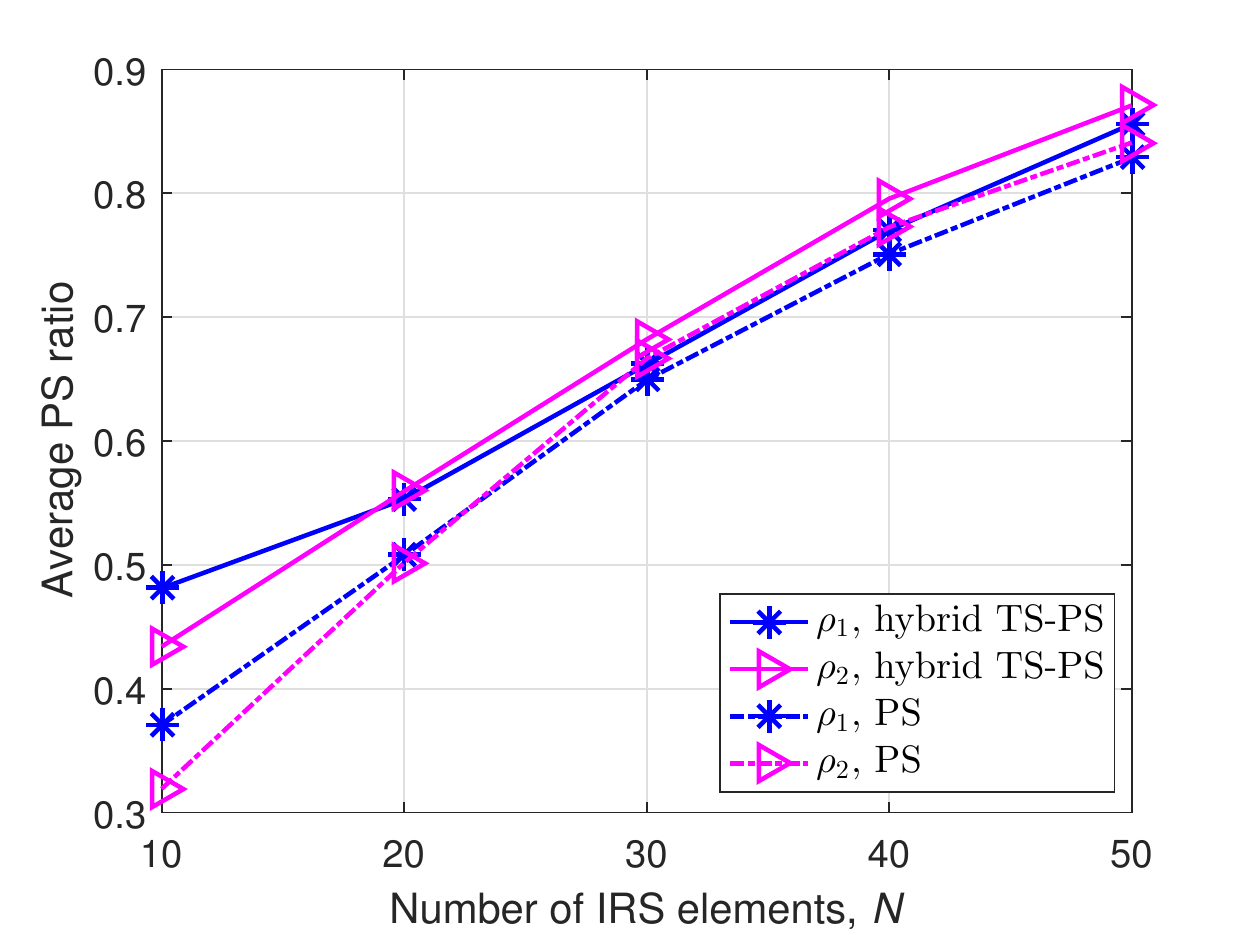}}
	\caption{(a) Average time fraction and (b) average PS ratio versus the number of IRS elements at the distributed IRSs for $d_{\rm T} = 0$ m, $M = K = 2$, and $E = 0.5$ $\mu\rm W$.}
	\label{fig:time_fraction_PS_ratio}
\end{figure}

To better understand how the hybrid TS-PS scheme outperforms the PS and TS schemes, we plot the average time fraction and the average PS ratio versus the number of IRS elements at the distributed IRSs when $d_{\rm T} = 0$ m in Figs. \ref{fig:time_fraction_vs_N} and \ref{fig:PS_ratio_vs_N}, respectively. Fig. \ref{fig:time_fraction_vs_N} shows that for the TS scheme, $\tau_1$ ($\tau_2$) decreases (increases) rapidly with the increase of $N$. This is because increasing $N$ can enhance the received signal power at each Rx \cite{2019_Qingqing_Joint}, thereby reducing the time required by each Tx to perform WPT to satisfy the EH requirement of its corresponding Rx. Consequently, more time becomes available for each Tx to engage in WIT, leading to an improvement in the overall sum rate performance. Nevertheless, with increasing $N$, the curve representing $\tau_1$ ($\tau_2$) of the hybrid TS-PS scheme only shows a slight downward (upward) trend. This implies that the increase in the achievable sum rate is mainly due to the increase in the PS ratios, which is confirmed by the numerical results shown in Fig. \ref{fig:PS_ratio_vs_N}. 
We also observe from Fig. \ref{fig:time_fraction_vs_N} that the 3rd time slot of the hybrid TS-PS scheme is allocated with a zero fraction of time. A possible explanation is that the Rxs can utilize the received signals more flexibly by adjusting the PS ratios in the 2nd time slot than in the 3rd time slot. As such, it may be unnecessary to allocate any time for the 3rd time slot. Moreover, compared to the TS scheme, the hybrid TS-PS scheme leaves more time for ID, although with part of the received signals being used for EH. Lastly, from Figs. \ref{fig:time_fraction_vs_N} and \ref{fig:PS_ratio_vs_N}, we can see that since the Rxs under the hybrid TS-PS scheme have harvested some energy by consuming a fraction of dedicated EH time, they can split more power for ID during the remaining time, as compared to the Rxs under the PS scheme. 

\begin{figure}[!t]
	\centering
	\subfigure[$K = 2$, $E = 0.6$ $\mu\rm W$.]{\label{fig:rate_vs_dTx_K2}
		\includegraphics[width=0.46\textwidth]{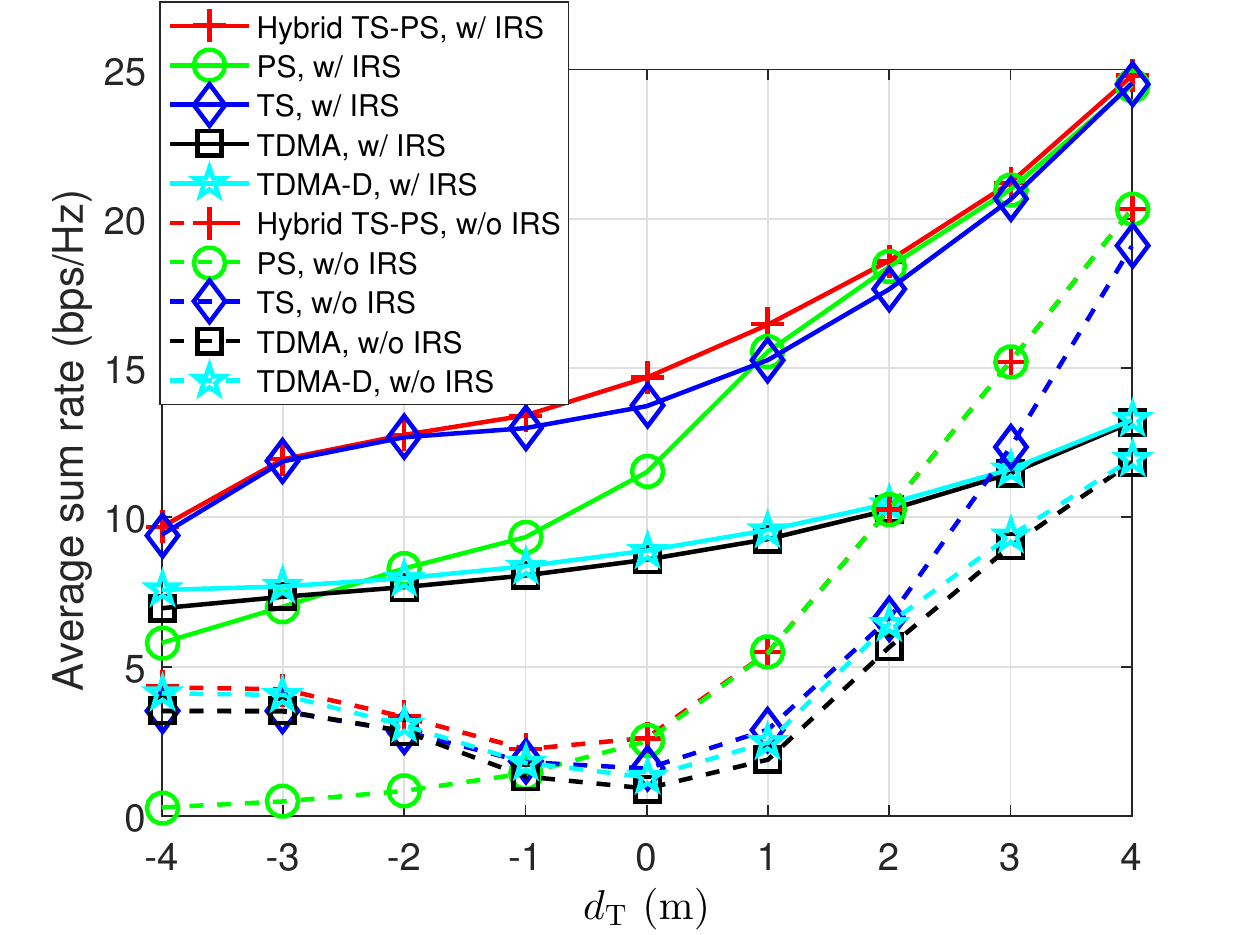}}
	\subfigure[$K = 5$, $E = 1.5$ $\mu\rm W$.]{\label{fig:rate_vs_dTx_K5}
		\includegraphics[width=0.46\textwidth]{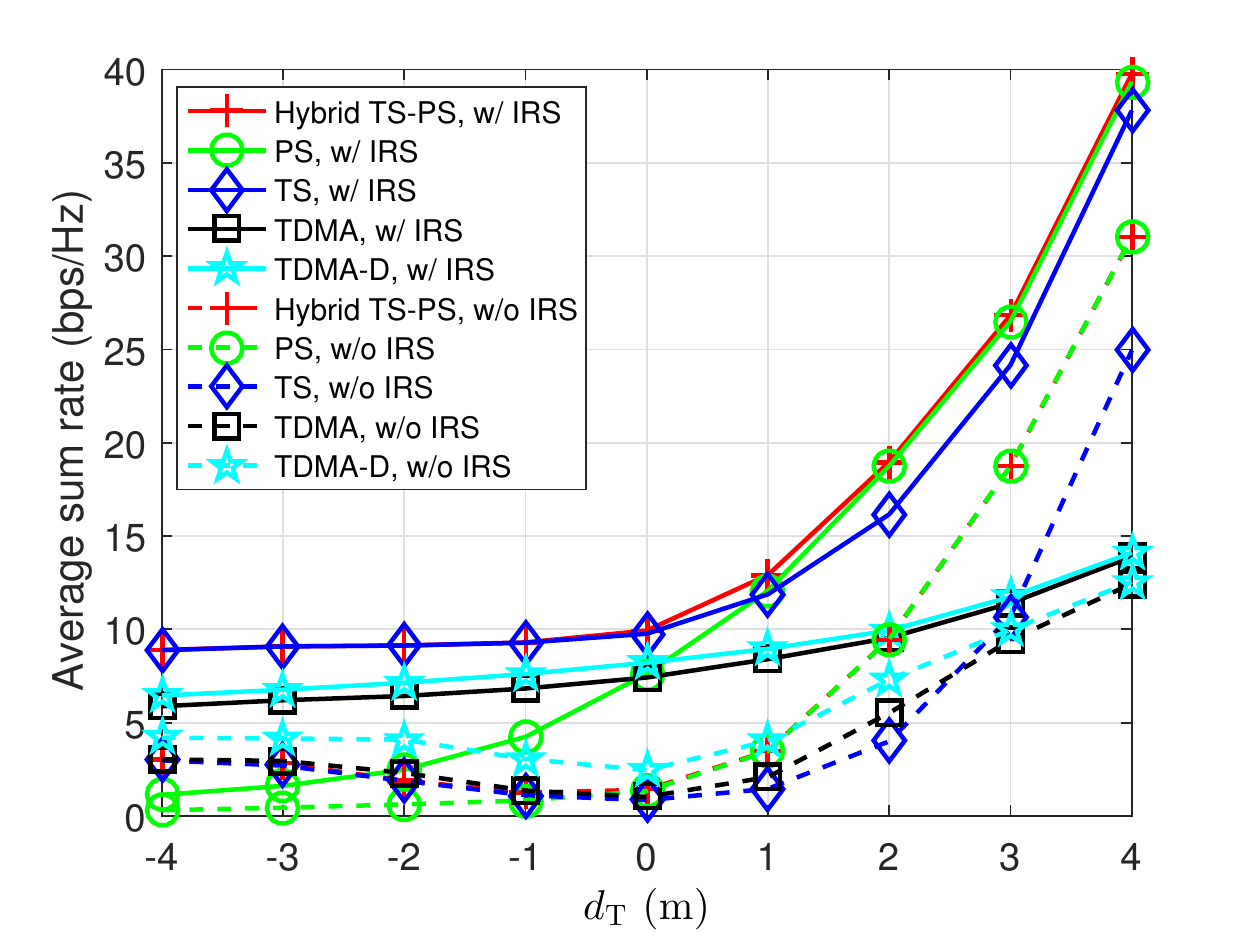}}
	\caption{Average sum rate versus $d_{\rm T}$ for $M = 2$ and $N = 40$.}
	\label{fig:rate_vs_dTx}
\end{figure}

\subsection{Impact of Cross-Link Channel Power}
In this subsection, we control the average relative cross-link channel power by varying $d_{\rm T}$ to study the impact of interference on the system performance. When $d_{\rm T} < 0$, the spherical polar coordinate of the $k$-th Tx is $\left( -d_{\rm T}, \frac{2\pi(k-1)}{K} + \pi, \frac{\pi}{2} \right)$ in m. Fig. \ref{fig:rate_vs_dTx} depicts the average sum rate versus $d_{\rm T}$ for different values of $K$ and $E$ when $M = 2$ and $N = 40$. Motivated by Fig. \ref{fig:rate_vs_N_dTx2}, we consider the distributed IRS deployment. The case without IRSs is also considered. From Fig. \ref{fig:rate_vs_dTx_K2}, it is first observed that for both cases with and without IRSs, the hybrid TS-PS scheme outperforms all the other schemes. Second, the introduction of IRSs widens the performance gap between the hybrid TS-PS and TDMA-based schemes, which is consistent with the observation in Fig. \ref{fig:rate_vs_N}. Third, when $d_{\rm T}$ increases from $0$ to $4$ m, all the schemes with and without IRSs experience an increase in the sum rate performance since the increased channel power of the direct links and reflected links (if any) favors both EH and ID. Fourth, when $d_{\rm T}$ decreases from $0$ to $-4$ m, the sum rates of all the other schemes except the PS scheme increase in the case without IRSs. The reasons are twofold. For one thing, the cross-link channel power becomes stronger with decreasing $d_{\rm T} < 0$, which is beneficial for EH. For another, the increased cross-link channel power has a direct negative impact on ID for the PS scheme, while the other schemes have the additional DoF in time to strike a balance between ID and EH for achieving better performance. In contrast, for the case with IRSs, all the rate curves decrease with decreasing $d_{\rm T}$. The explanation is that when deploying IRSs, the small EH demand can be easily satisfied and thus the increased cross-link channel power only causes rate reduction to all the IRS-aided schemes. 

The majority of the above observations can also be observed in Fig. \ref{fig:rate_vs_dTx_K5} when $K$ and $E$ are increased to $5$ and $1.5$ $\mu\rm W$, respectively. What is different is that in the absence of IRSs, the performance of the hybrid TS-PS scheme is inferior to that of the TDMA-D scheme when $d_{\rm T} < 2$ m. This is mainly because, under the setting of $K > M$, the hybrid TS-PS scheme cannot harness the strong cross-link interference well, while the TDMA-D scheme is free from the cross-link interference. Nevertheless, the advantages of the TDMA-D scheme will be weakened or even vanish after the introduction of IRSs, since IRSs can help the hybrid TS-PS scheme to better mitigate interference. Partly as a result, the hybrid TS-PS scheme outperforms the TDMA-D scheme in the case with IRSs. \looseness=-1 

\begin{figure*}[!t]
	\subfigcapskip = 6pt 
	\hspace{-2mm}
	\subfigure[$K = 2$.]{
		\begin{minipage}[]{0.33\textwidth}\label{fig:rate_vs_E_M4K2}
			\includegraphics[scale=0.45]{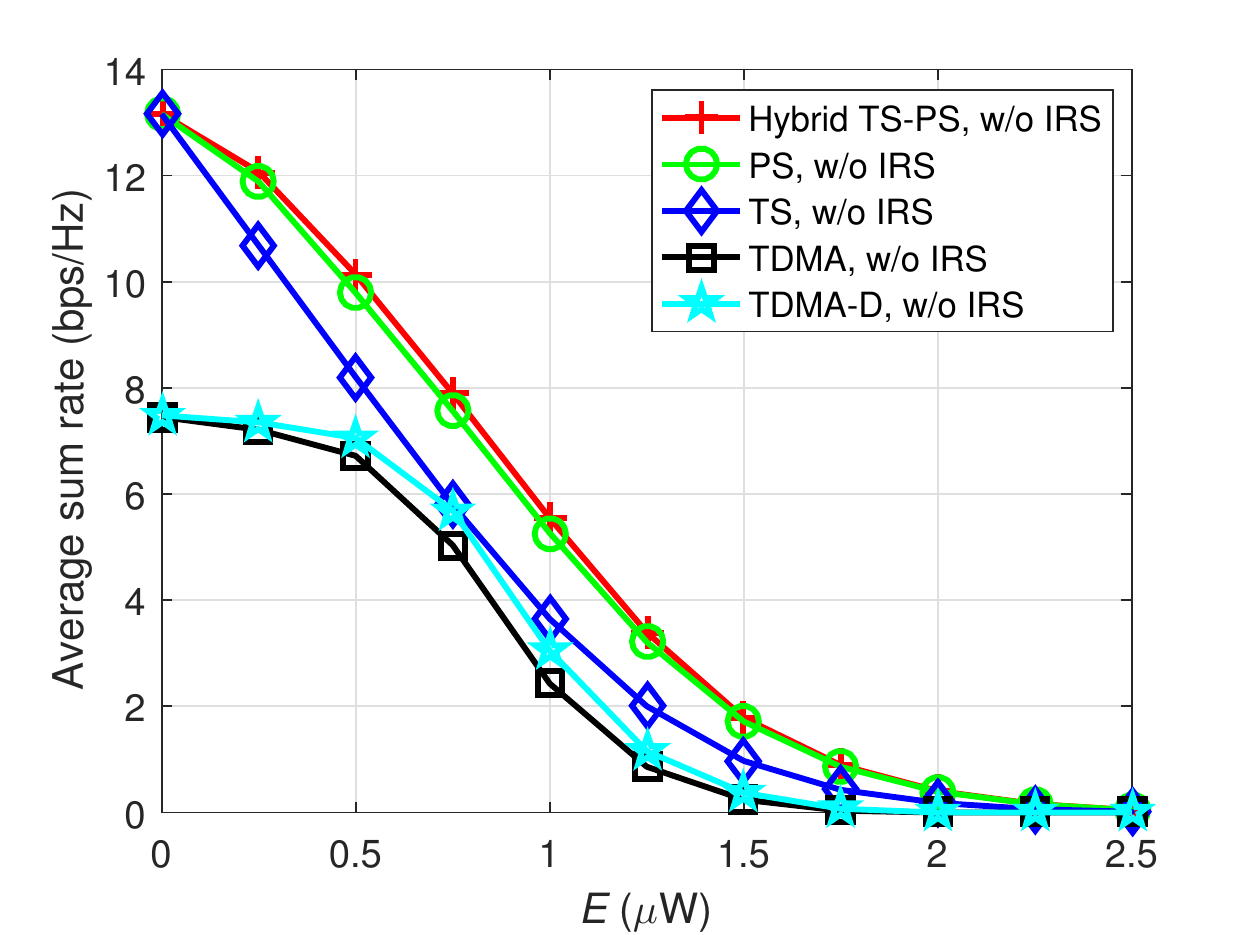}  \\ 
			\includegraphics[scale=0.45]{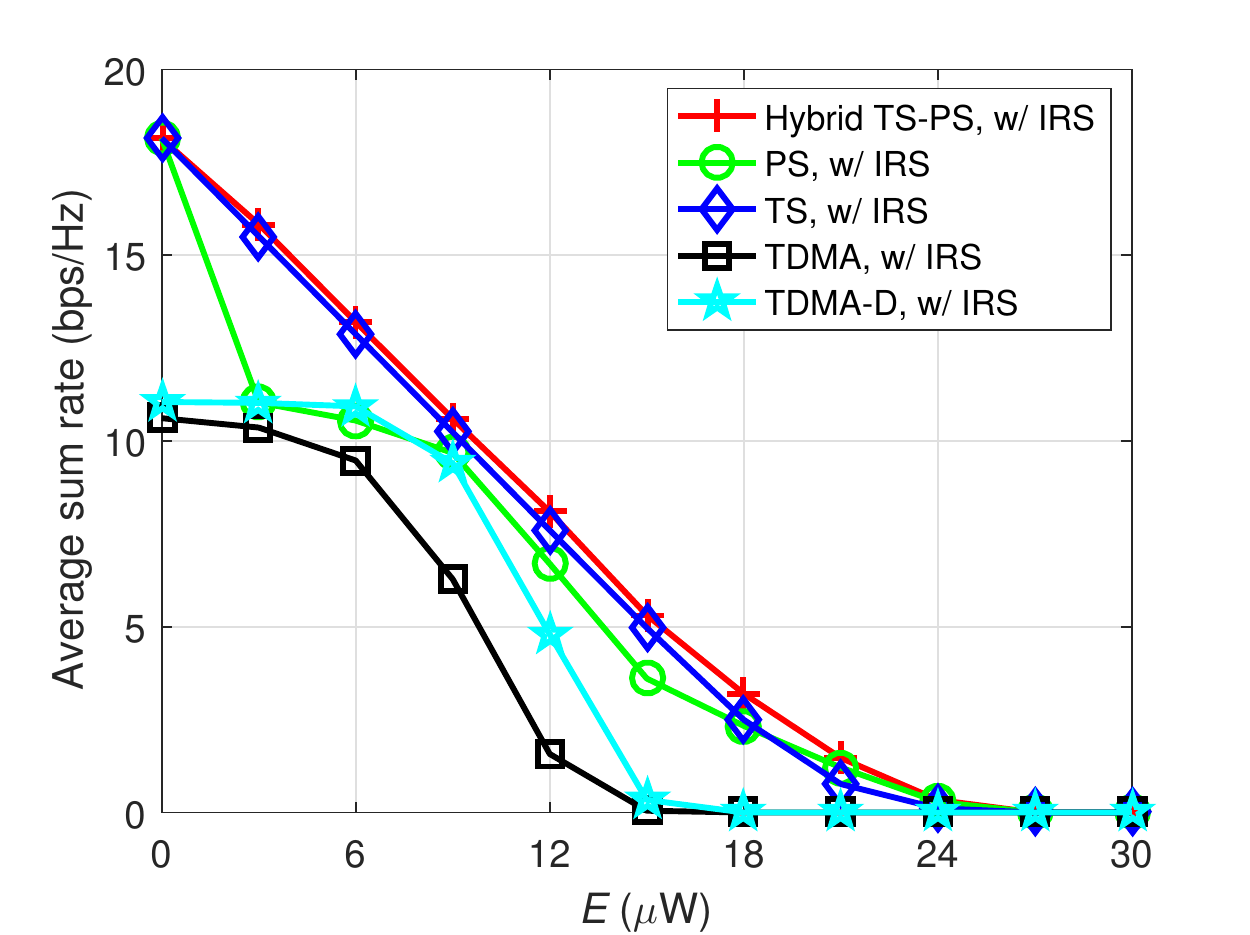}     
	\end{minipage}}
    \hspace{-5mm}
	\subfigure[$K = 4$.]{
		\begin{minipage}[]{0.33\textwidth}\label{fig:rate_vs_E_M4K4}
			\includegraphics[scale=0.45]{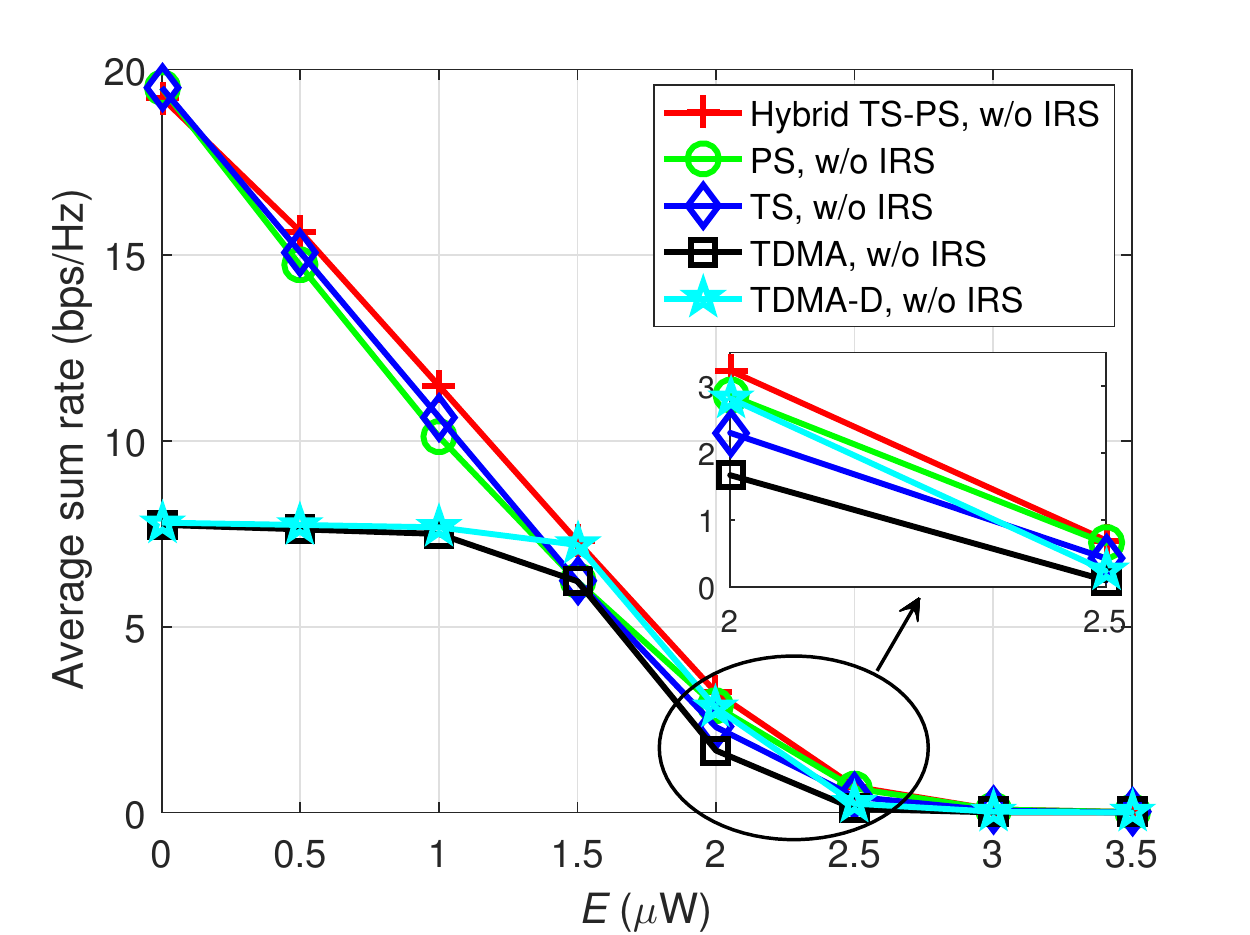}  \\ 
			\includegraphics[scale=0.45]{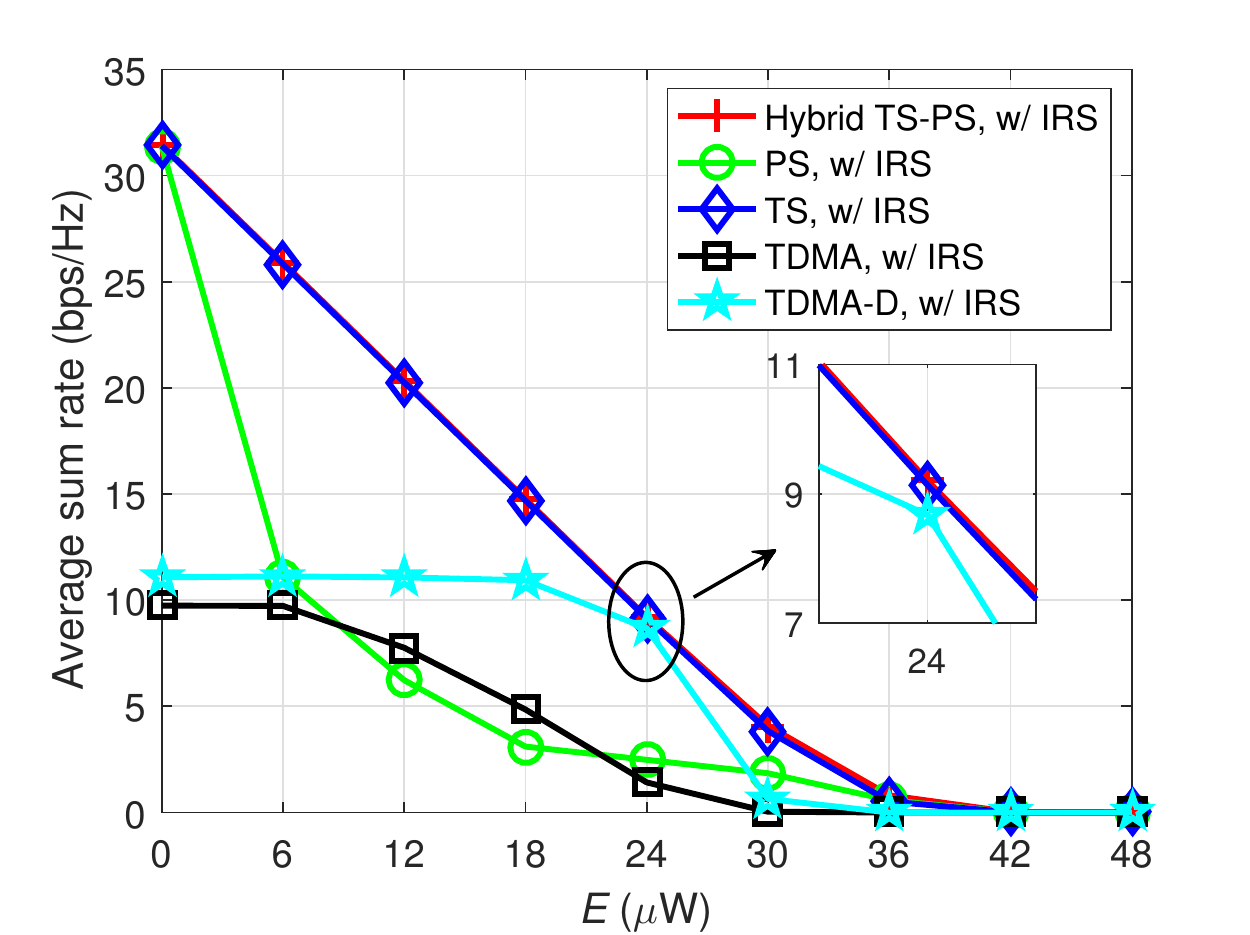}     
	\end{minipage}}
    \hspace{-5mm}
	\subfigure[$K = 8$.]{
		\begin{minipage}[]{0.33\textwidth}\label{fig:rate_vs_E_M4K8}
			\includegraphics[scale=0.45]{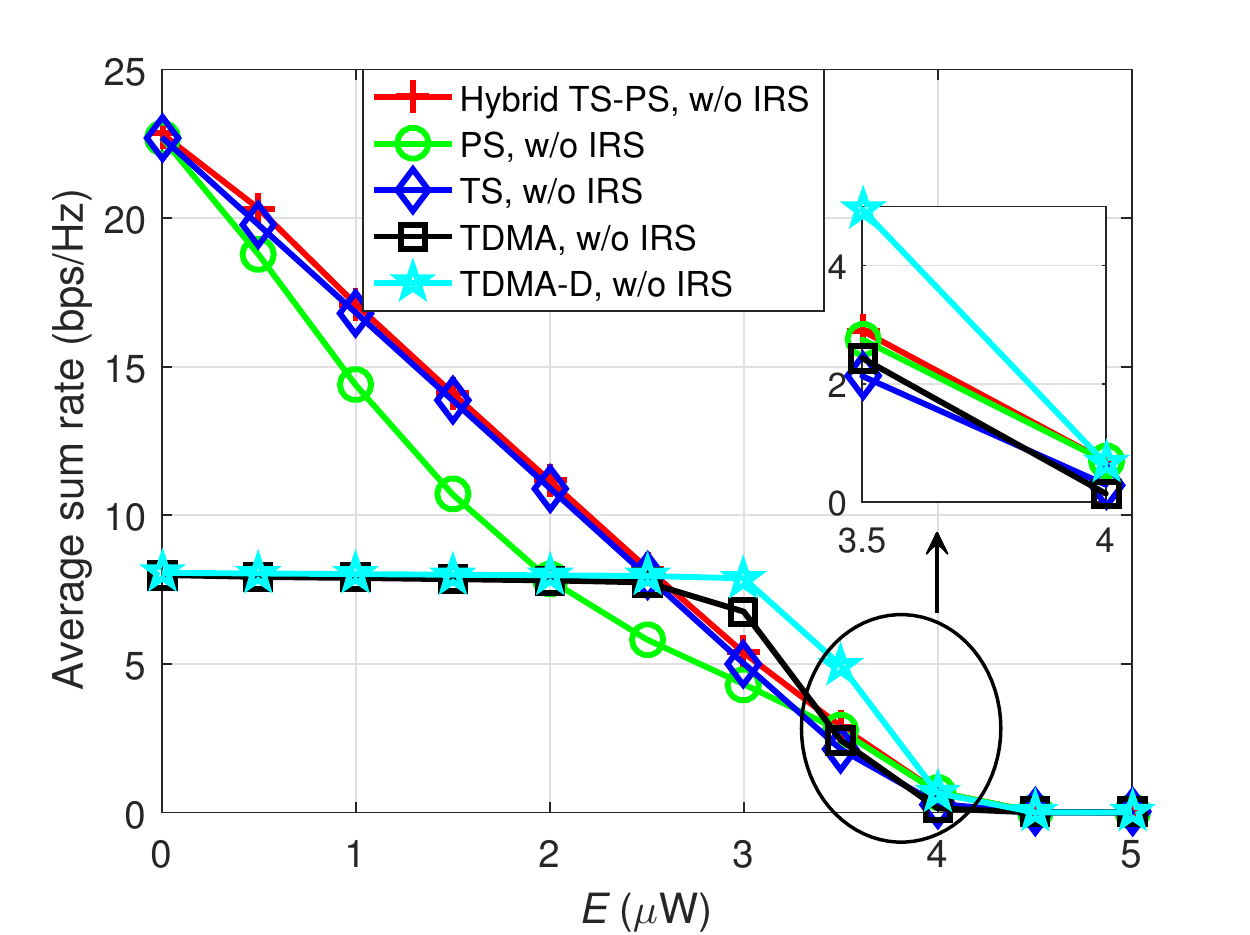}  \\ 
			\includegraphics[scale=0.45]{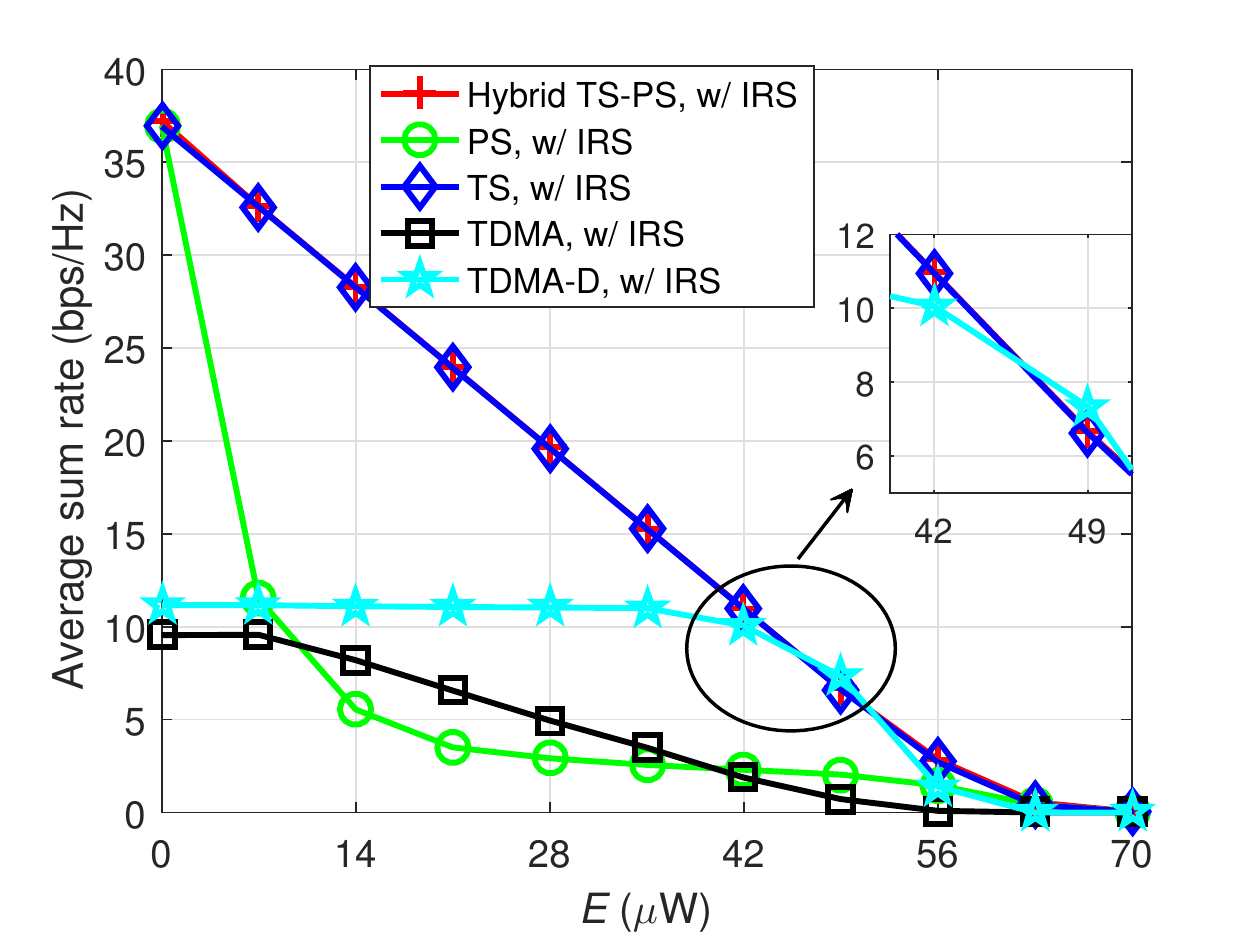}     
	\end{minipage}}
	\caption{Average sum rate versus $E$ for $d_{\rm T} = 0$ m, $M = 4$, and $N = 40$.}
	\label{fig:rate_vs_E_K}
\end{figure*}

\subsection{Impact of EH Requirement and Number of Tx-Rx Pairs}
In the final set of numerical experiments, we investigate the impact of the EH requirement and number of Tx-Rx pairs on the performance comparison. Specifically, Fig. \ref{fig:rate_vs_E_K} depicts the average sum rate versus $E$ for $d_{\rm T} = 0$ m, $M = 4$, $N = 40$, and $K\in\{2,4,8\}$. Two cases with and without the centralized IRS are considered. From Figs. \ref{fig:rate_vs_E_M4K2}-\ref{fig:rate_vs_E_M4K8}, it is first observed that under different values of $K$, the IRS-aided designs are superior to their counterparts without IRS, in terms of both the maximum EH requirement that can be satisfied and the achievable sum rate. This is expected since the high passive beamforming gain promised by the IRS can improve the efficiency of both WPT and WIT. Second, it can be seen that without IRS, the TS scheme is superior to the PS scheme when $K \geq M$ and $E$ is not very large. However, after deploying an IRS, the TS scheme outperforms the PS scheme for almost all the considered values of $K$ and $E$. The reason is that compared to the PS scheme, the TS scheme allows the IRS to reconfigure its phase-shift vector twice to achieve more flexible resource allocation. Finally, we observe from Figs. \ref{fig:rate_vs_E_M4K2} and \ref{fig:rate_vs_E_M4K4} that when $K \leq M$, the hybrid TS-PS scheme achieves the highest sum rate in all the feasible $E$ regimes. On the other hand, Fig. \ref{fig:rate_vs_E_M4K8} shows that when $K > M$ and $E$ is relatively large, the TDMA-D scheme may outperform the hybrid TS-PS scheme. Nevertheless, by comparing the two figures in Fig. \ref{fig:rate_vs_E_M4K8}, we observe that when $N$ increases from $0$ to $40$, the probability of the TDMA-D scheme performing better than the hybrid TS-PS scheme is greatly reduced. This observation motivates us to further increase $N$, as shown in Fig. \ref{fig:rate_vs_E_N}. It can be seen from Fig. \ref{fig:rate_vs_E_N} that if the IRS-aided hybrid TS-PS scheme is inferior to the IRS-aided TDMA-D scheme in some scenarios where the interference is overwhelming and the EH requirement is also stringent, we can reverse this result by further increasing $N$. The reasons arise from two aspects. First, a larger $N$ can help the hybrid TS-PS scheme to better mitigate interference, which further weakens the advantage of the interference-free TDMA-D scheme. Second, increasing $N$ can enable the Txs adopting the hybrid TS-PS scheme to shorten the time spent on WPT and leave more time for WIT, but it can hardly extend the short information transmission durations of the Txs adopting the TDMA-D scheme caused by the time division.   

\begin{figure}[!t]
	\centering
	\includegraphics[width=0.46\textwidth]{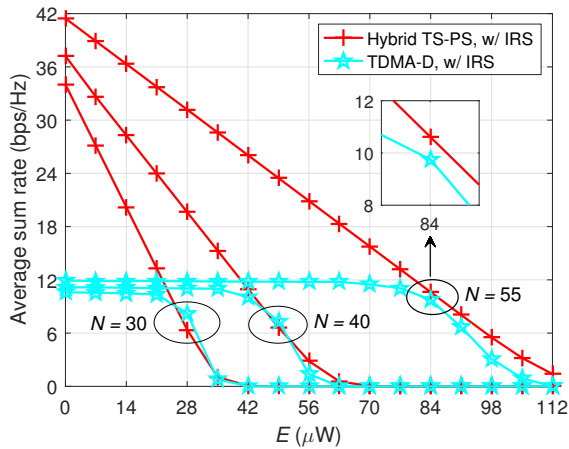}
	\caption{Average sum rate versus $E$ for $d_{\rm T} = 0$ m, $M = 4$, and $K = 8$.}
	\label{fig:rate_vs_E_N}
\end{figure}

\section{Conclusion}\label{Sec_conclusion}
In this paper, we proposed three practical transmission schemes, namely the IRS-aided hybrid TS-PS, TDMA, and TDMA-D schemes, for SWIPT in a multi-user MISO IFC. To maximize the sum rate of the Rxs while satisfying their individual EH requirements, we formulated three optimization problems, each corresponding to a proposed scheme, where the transmit covariance matrices, the IRS phase-shift vectors, and the resource allocation were jointly optimized. Since these optimization problems with the coupled variables were non-convex, we developed AO-based algorithms to solve them suboptimally. Simulation results verified the effectiveness of our proposed algorithms and provided some meaningful insights. First, the distributed IRS deployment is preferred for scenarios where the Tx-Rx pairs are geographically dispersed, while the centralized IRS deployment is a better choice for the opposite scenarios. Second, the performance comparison results among the schemes can vary significantly before and after introducing the IRSs. Third, the IRS-aided schemes can considerably enhance the performance of both WIT and WPT compared to the benchmark schemes without IRSs. Finally, the IRS-aided hybrid TS-PS scheme generally outperforms all other considered schemes in terms of sum rate performance. If not, increasing the number of IRS elements can achieve this. 

It is worth pointing out that since these formulated non-convex problems can hardly be solved by existing optimization techniques optimally and analytically, in this paper, the performance comparison among different schemes is limited to the simulation results obtained by the proposed suboptimal algorithms. How to derive theoretical results to validate or invalidate this paper's simulation results is a challenging issue and needs further study. Besides, in addition to our considered schemes, other transmission schemes certainly exist to separate the EH and ID modes of the Rxs across time. It will be interesting to see how their performance compares to that of our considered schemes. 

\appendix[Proof of Proposition \ref{prop1}]
First, it should be noted that the constraints in \eqref{P1-sub1-Eqv1_cons:d} associated with $\tau_1$ must be active at the optimal solution to problem \eqref{P1-sub1-Eqv1-SCA}, i.e., ${\rm tr}\left(\mathbf W_{i,1}^{\star}\right) =\tau_1^{\star}P_i$, $\forall i\in\mathcal K$. Then, to facilitate the proof and without loss of optimality and generality, we transform problem \eqref{P1-sub1-Eqv1-SCA} into the following equivalent form: 
\begin{subequations}\label{P1-sub1-Eqv1-SCA-Eqv}
	\begin{align}
	&\hspace{-9mm}\underset{\substack{\{\mathbf W_{i,j} \succeq \mathbf 0\}, \{\tau_j\}, \{\rho_k\}, \\ \{e_k\}, \{z_k\}, i,k\in\mathcal K, j\in\{1,2,3\}}}{\min} \hspace{1mm} \sum_{k=1}^K\sum_{j=2}^3 \left[ -\tau_j\log_2\left(\frac{\delta_{k,j}}{\tau_j}\right)\right.  \nonumber\\
	& \hspace{2.5cm}\left. + \frac{\sum\nolimits_{i=1,i\neq k}^K{\rm tr}\left(\mathbf a_{i,k,j}\mathbf a_{i,k,j}^H\mathbf W_{i,j}\right)}{\Psi_{k,j}^t\ln2}\right] \nonumber\\ 
	&\hspace{2.5cm} + \Pi \\
	\hspace{4mm}\text{s.t.} \hspace{3mm}& \delta_{k,2} \leq \sum_{i=1}^K{\rm tr}\left(\mathbf a_{i,k,2}\mathbf a_{i,k,2}^H\mathbf W_{i,2}\right) + e_k\hat\sigma_k^2 + \tau_2\tilde\sigma_k^2, \nonumber\\
	& \forall k\in\mathcal K, \label{P1-sub1-Eqv1-SCA-Eqv_cons:b}\\
	& \delta_{k,3} \leq \sum_{i=1}^K{\rm tr}\left(\mathbf a_{i,k,3}\mathbf a_{i,k,3}^H \mathbf W_{i,3}\right) + \tau_3\sigma_k^2, \ \forall k\in\mathcal K, \label{P1-sub1-Eqv1-SCA-Eqv_cons:c}\\
	& \eqref{P1-sub1-Eqv1_cons:b_sca}, \eqref{P1-sub1-Eqv1_cons:c_1_sca}, \eqref{P1-sub1-Eqv1_cons:c_2}, \eqref{P1-sub1-Eqv1_cons:d}, \eqref{P1_cons:d}, \eqref{P1_cons:e},
	\end{align}
\end{subequations}	
where $\{\delta_{k,j}\}$ are newly introduced slack variables and $\Pi$ denotes the collection of all the terms that do not involve $\{\mathbf W_{i,j}\}$, $k,i\in\mathcal K$, $j\in\{2,3\}$. Problems \eqref{P1-sub1-Eqv1-SCA} and \eqref{P1-sub1-Eqv1-SCA-Eqv} are equivalent since it can be shown by contradiction that the constraints in \eqref{P1-sub1-Eqv1-SCA-Eqv_cons:b} and \eqref{P1-sub1-Eqv1-SCA-Eqv_cons:c} hold with equality at the optimum. Furthermore, it can be verified that the convex problem \eqref{P1-sub1-Eqv1-SCA-Eqv} satisfies the Slater's condition and thus enjoys a zero duality gap \cite{2014_Jie_SWIPT}. Then, we consider the Lagrangian of problem \eqref{P1-sub1-Eqv1-SCA-Eqv} given by 
	\begin{align}
	\mathcal L^{\rm Hy} & = \Xi +  \sum_{k=1}^K\sum_{j=2}^3\frac{\sum\nolimits_{i=1,i\neq k}^K{\rm tr}\left(\mathbf a_{i,k,j}\mathbf a_{i,k,j}^H\mathbf W_{i,j}\right)}{\Psi_{k,j}^t\ln2} \nonumber\\
	& - \sum_{k=1}^K\sum_{j=2}^3\alpha_{k,j}\sum_{i=1}^K{\rm tr}\left(\mathbf a_{i,k,j}\mathbf a_{i,k,j}^H\mathbf W_{i,j}\right) \nonumber\\ 
	& - \sum_{k=1}^K\sum_{j=1}^2\nu_{k,j}\zeta\sum_{i=1}^K{\rm tr}\left(\mathbf a_{i,k,j}\mathbf a_{i,k,j}^H \mathbf W_{i,j}\right) \nonumber\\
	& + \sum_{i=1}^K\sum_{j=1}^3\lambda_{i,j}{\rm tr}\left(\mathbf W_{i,j}\right) - \sum_{i=1}^K\sum_{j=1}^3{\rm tr}\left(\mathbf W_{i,j}\mathbf T_{i,j}\right),
	\end{align}
where the non-negative variables $\{\alpha_{k,2}\}$, $\{\alpha_{k,3}\}$, $\{\nu_{k,1}\}$, $\{\nu_{k,2}\}$, $\{\lambda_{i,j}\}$, and $\{\mathbf T_{i,j} \succeq \mathbf 0\}$, $i,k\in\mathcal K$, $j\in\{1,2,3\}$ are the dual variables associated with constraints \eqref{P1-sub1-Eqv1-SCA-Eqv_cons:b}, \eqref{P1-sub1-Eqv1-SCA-Eqv_cons:c},  \eqref{P1-sub1-Eqv1_cons:c_1_sca}, \eqref{P1-sub1-Eqv1_cons:c_2}, \eqref{P1-sub1-Eqv1_cons:d}, and $\mathbf W_{i,j} \succeq \mathbf 0$, respectively, and $\Xi$ includes all the terms that do not involve $\{\mathbf W_{i,j}\}$. According to the Karush-Kuhn-Tucker (KKT) conditions, the optimal $\{\mathbf W_{i,1}^{\star}\}$ satisfy
	\begin{align}
	& \text{K1:} \ \nu_{k,1}^{\star}, \lambda_{i,1}^{\star} \geq 0, \ \mathbf T_{i,1}^{\star}\succeq \mathbf 0, \ \forall i,k\in\mathcal K,\\
	& \text{K2:} \ \mathbf T_{i,1}^{\star}\mathbf W_{i,1}^{\star} = \mathbf 0, \ \forall i\in\mathcal K, \\
	& \text{K3:} \ \mathbf T_{i,1}^{\star} =  \lambda_{i,1}^{\star}\mathbf I_M - \zeta\sum_{k=1}^K\nu_{k,1}^{\star}\mathbf a_{i,k,1}\mathbf a_{i,k,1}^H,  \ \forall i\in\mathcal K,
	\end{align}
where $\{\nu_{k,1}^{\star}\}$, $\{\lambda_{i,1}^{\star}\}$, and $\{\mathbf T_{i,1}^{\star}\}$ are optimal for the dual problem of \eqref{P1-sub1-Eqv1-SCA-Eqv}, and the equalities in K3 are derived from $\nabla_{\mathbf W_{i,1}^{\star}}\mathcal L^{\rm Hy} = \mathbf 0$, $\forall i\in\mathcal K$.

Next, we prove the condition C1 by exploring the structure of $\{\mathbf T_{i,1}^{\star}\}$. First, since $\tau_1^{\star} > 0$ and the constraints in \eqref{P1-sub1-Eqv1_cons:d} associated with $\tau_1^{\star}$ are active, the optimal dual variables corresponding to these constraints are positive, i.e., $\lambda_{i,1}^{\star} > 0$, $\forall i\in\mathcal K$. Second, $\mathbf \Upsilon_{i,1} \triangleq \zeta\sum_{k=1}^K\nu_{k,1}^{\star}\mathbf a_{i,k,1}\mathbf a_{i,k,1}^H$,  $\forall i\in\mathcal K$, cannot be negative semidefinite. This can be proved by contradiction: if $\mathbf \Upsilon_{i,1} \prec \mathbf 0$ for $i\in\mathcal K$, then $\mathbf T_{i,1}^{\star}$ becomes a full-rank matrix due to K3. This together with K2 forces $\mathbf W_{i,1}^{\star}$ to be a zero matrix which cannot be optimal to problem \eqref{P1-sub1-Eqv1-SCA-Eqv} when $E_k > 0$, $\forall k\in\mathcal K$ and $\tau_1^{\star} > 0$. Thus, we only need to consider the case of $\mathbf \Upsilon_{i,1} \succeq \mathbf 0$, $\forall i\in\mathcal K$. In this case, it can be verified from K1 and K3 that $\lambda_{i,1}^{\star} \geq \lambda_{\mathbf \Upsilon_{i,1}}^{\rm max} \geq 0$, $\forall i\in\mathcal K$, where $\lambda_{\mathbf \Upsilon_{i,1}}^{\rm max}$ denotes the dominant eigenvalue of $\mathbf \Upsilon_{i,1}$. If $\lambda_{i,1}^{\star} > \lambda_{\mathbf \Upsilon_{i,1}}^{\rm max}$ for $i\in\mathcal K$, then $\mathbf T_{i,1}^{\star}$ is of full rank which has been contradicted. Therefore, it must hold that $\lambda_{i,1}^{\star} = \lambda_{\mathbf \Upsilon_{i,1}}^{\rm max}$, $\forall i\in\mathcal K$. On the other hand, the possibility that multiple
eigenvalues of $\mathbf \Upsilon_{i,1}$, $i\in\mathcal K$ have the same value $\lambda_{\mathbf \Upsilon_{i,1}}^{\rm max}$ is zero since $\{\mathbf a_{i,k,1}\}_{\forall k\in\mathcal K}$ are assumed to be independently distributed. As a result, it follows that ${\rm rank}\left(\mathbf T_{i,1}^{\star}\right) = M - 1 $, $\forall i\in\mathcal K$. Moreover, according to Sylvester’s inequality \cite{2013_Howard_rank}, it must hold that ${\rm rank}\left(\mathbf T_{i,1}^{\star}\right) + {\rm rank}\left(\mathbf W_{i,1}^{\star}\right) - M \leq {\rm rank}\left(\mathbf T_{i,1}^{\star}\mathbf W_{i,1}^{\star}\right)$, $\forall i\in\mathcal K$. Due to the fact that ${\rm rank}\left(\mathbf T_{i,1}^{\star}\mathbf W_{i,1}^{\star}\right) = 0$ (see K2), we have ${\rm rank}\left(\mathbf W_{i,1}^{\star}\right) \leq 1$, $\forall i\in\mathcal K$. Since we have already clarified earlier that $\mathbf W_{i,1}^{\star} = \mathbf 0$, $i\in\mathcal K$ cannot be an optimal solution, it can be verified that ${\rm rank}\left(\mathbf W_{i,1}^{\star}\right) = 1$, which completes the proof of the condition C1. 

The proofs of the conditions C2 and C3 are similar to that of C1 provided above and are therefore omitted for brevity. 

\bibliographystyle{IEEEtran}
\bibliography{ref}

\begin{IEEEbiography}[{\includegraphics[width=1in,height=1.25in,clip,keepaspectratio]{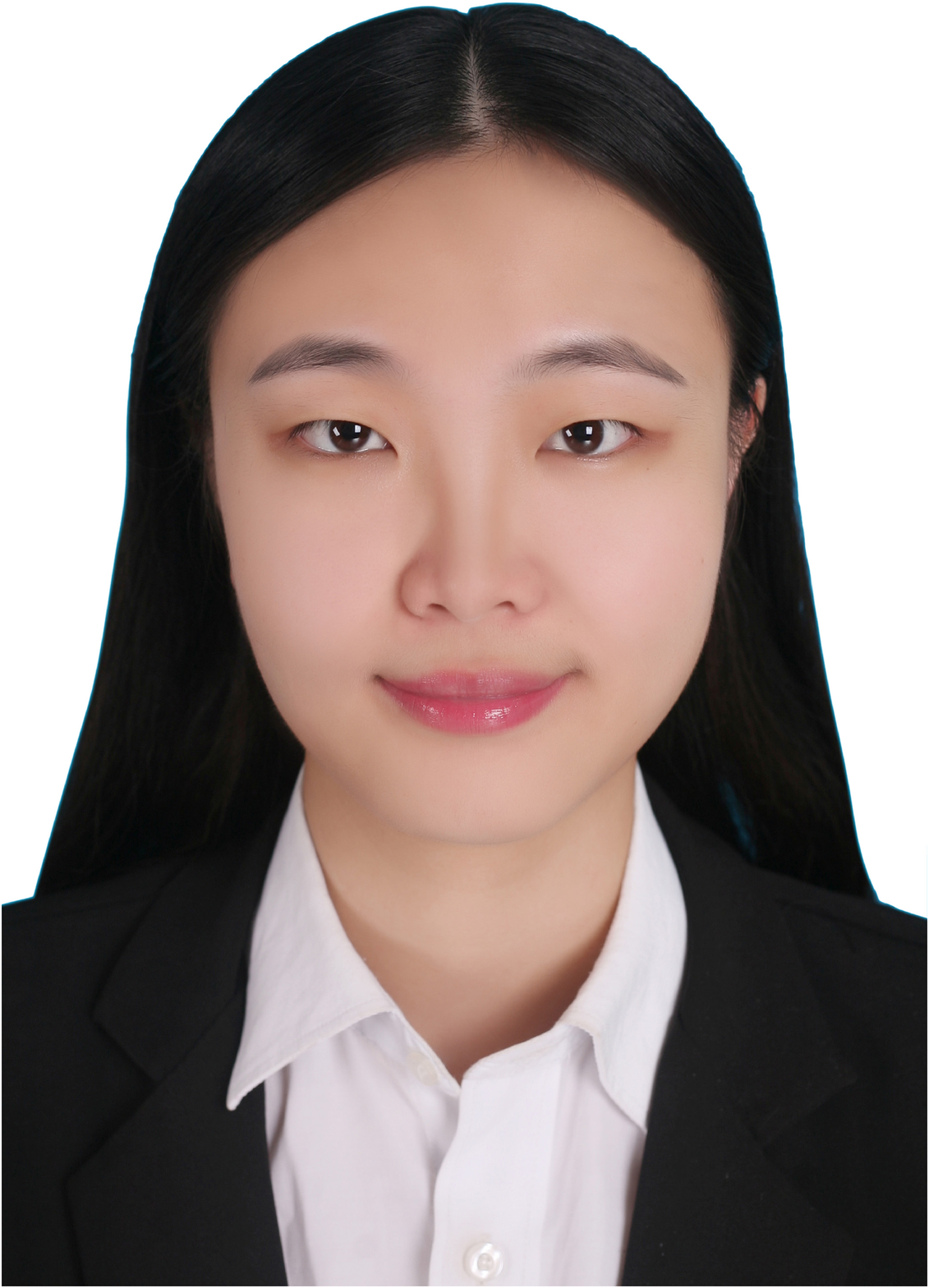}}]{Ying Gao} received the B.Eng. degree in electronic information engineering from Nanjing University of Science and Technology, Nanjing, China, in 2016, and the Ph.D. degree in communications and information systems from Shanghai Institute of Microsystem and Information Technology, Chinese Academy of Sciences, Shanghai, China, in 2021. She is currently a Post-Doctoral Researcher with the State Key Laboratory of Internet of Things for Smart City, University of Macau. Her current research interests include intelligent reflecting surface (IRS) assisted communications, unmanned aerial vehicle (UAV) enabled communications, physical layer security, and optimization theory. 
\end{IEEEbiography}

\begin{IEEEbiography}[{\includegraphics[width=1in,height=1.25in,clip,keepaspectratio]{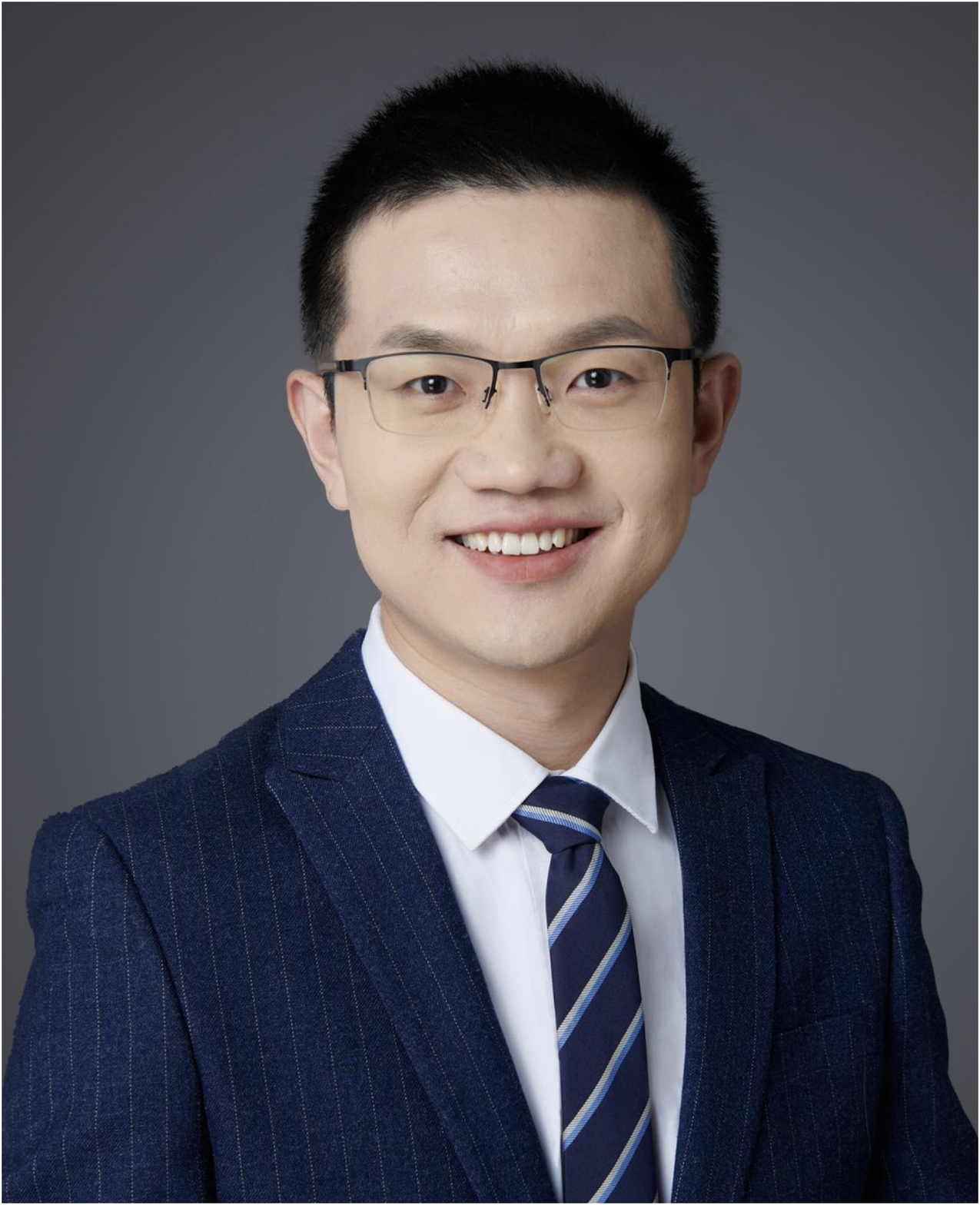}}]{Qingqing Wu} (Senior Member, IEEE) is an Associate Professor with Shanghai Jiao Tong University. His current research interest includes intelligent reflecting surface (IRS), unmanned aerial vehicle (UAV) communications, and MIMO transceiver design. He has coauthored more than 100 IEEE journal papers with 29 ESI highly cited papers and 9 ESI hot papers, which have received more than 22,000 Google citations. He was listed as the Clarivate ESI Highly Cited Researcher in 2022 and 2021, the Most Influential Scholar Award in AI-2000 by Aminer in 2021 and World’s Top 2$\%$ Scientist by Stanford University in 2020 and 2021. 
		
He was the recipient of the IEEE Communications Society Fred Ellersick Prize, IEEE  Best Tutorial Paper Award in 2023, Asia-Pacific Best Young Researcher Award and Outstanding Paper Award in 2022, Young Author Best Paper Award in 2021, the Outstanding Ph.D. Thesis Award of China Institute of Communications in 2017, the IEEE ICCC Best Paper Award in 2021, and IEEE WCSP Best Paper Award in 2015. He was the Exemplary Editor of IEEE Communications Letters in 2019 and the Exemplary Reviewer of several IEEE journals. He serves as an Associate Editor for IEEE Transactions on Communications, IEEE Communications Letters, and IEEE Wireless Communications Letters. He is the Lead Guest Editor for IEEE Journal on Selected Areas in Communications. He is the workshop co-chair for IEEE ICC 2019-2023 and IEEE GLOBECOM 2020. He serves as the Workshops and Symposia Officer of Reconfigurable Intelligent Surfaces Emerging Technology Initiative and Research Blog Officer of Aerial Communications Emerging Technology Initiative. He is the IEEE Communications Society Young Professional Chair in Asia Pacific Region.
\end{IEEEbiography}

\begin{IEEEbiography}[{\includegraphics[width=1in,height=1.25in,clip,keepaspectratio]{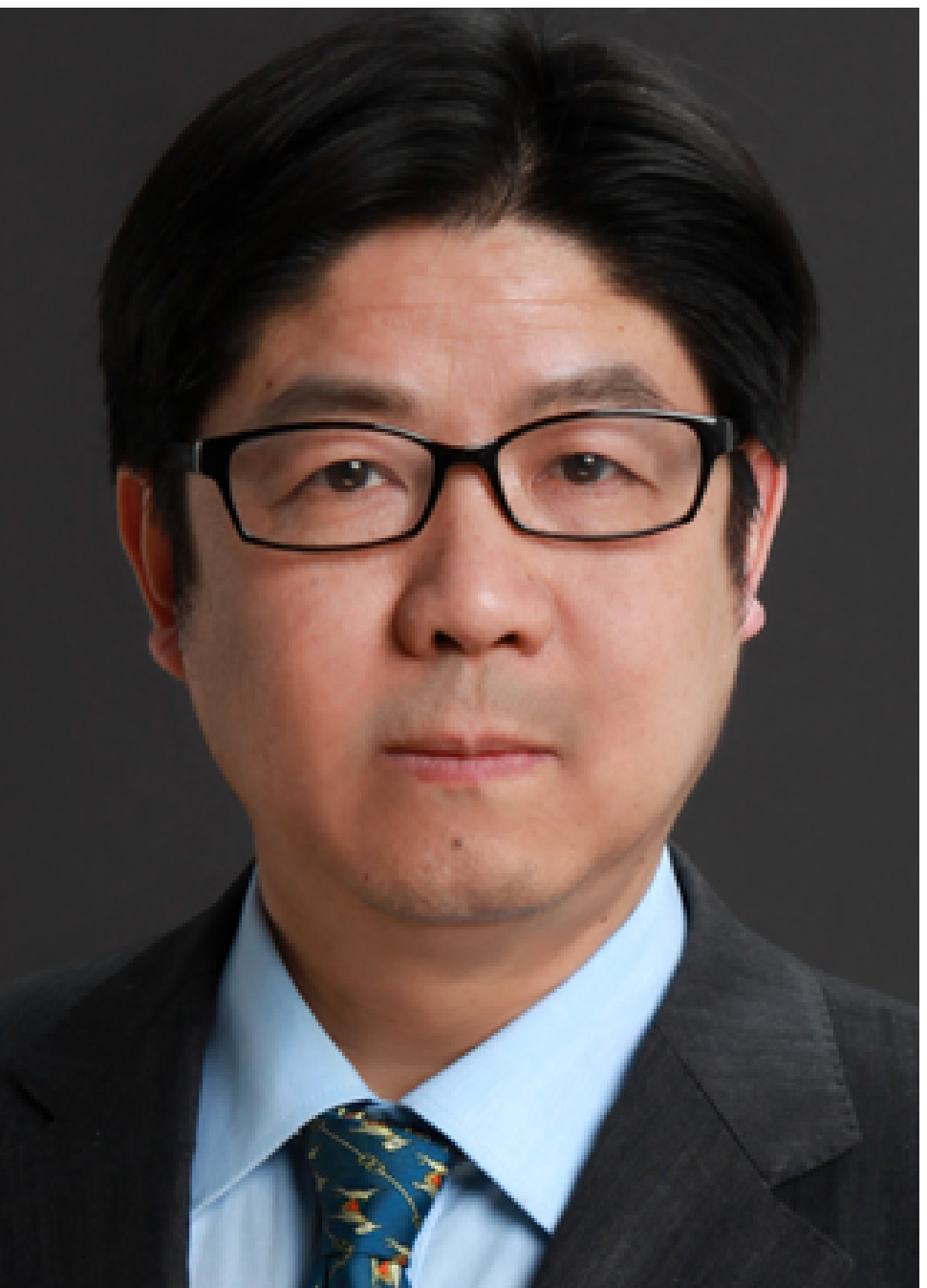}}] {Wen Chen} (Senior Member, IEEE) is a tenured Professor with the Department of Electronic Engineering, Shanghai Jiao Tong University, China, where he is the director of Broadband Access Network Laboratory. He is a fellow of Chinese Institute of Electronics and the distinguished lecturers of IEEE Communications Society and IEEE Vehicular Technology Society. He is the Shanghai Chapter Chair of IEEE Vehicular Technology Society, Editors of IEEE Transactions on Wireless Communications, IEEE Transactions on Communications, IEEE Access and IEEE Open Journal of Vehicular Technology. His research interests include multiple access, wireless AI and meta-surface communications. He has published more than 130 papers in IEEE journals and more than 120 papers in IEEE Conferences, with citations more than 9000 in google scholar.
\end{IEEEbiography}	

\begin{IEEEbiography}[{\includegraphics[width=1in,height=1.25in,clip,keepaspectratio]{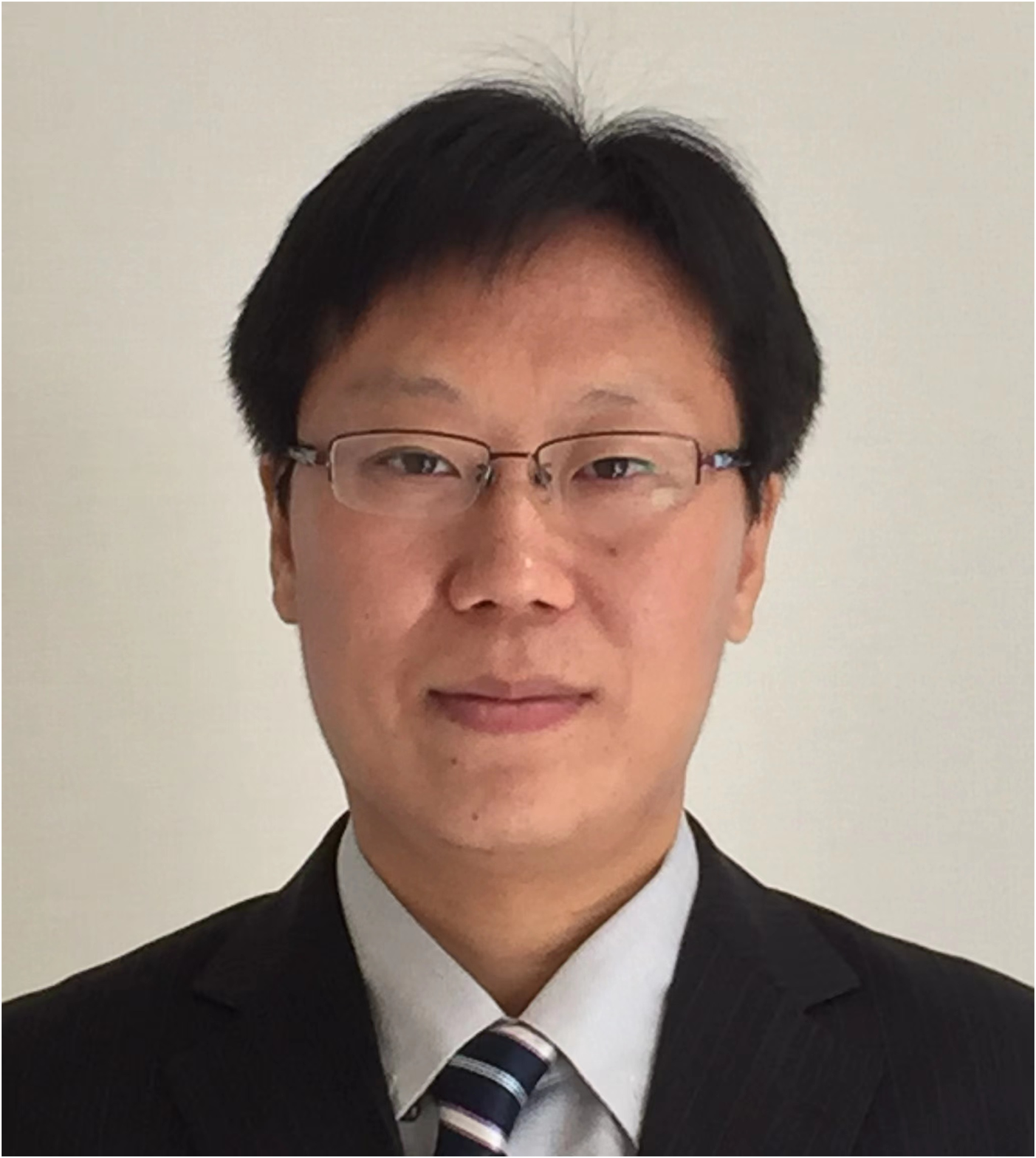}}]{Celimuge Wu} (Senior Member, IEEE) received his Ph.D. degree from The University of Electro-Communications, Japan. He is currently a professor and the director of Meta-Networking Research Center, The University of Electro-Communications. His research interests include Vehicular Networks, Edge Computing, IoT, and AI for Wireless Networking and Computing. He serves as an associate editor of IEEE Transactions on Cognitive Communications and Networking, IEEE Transactions on Network Science and Engineering, and IEEE Transactions on Green Communications and Networking. He is Vice Chair (Asia Pacific) of IEEE Technical Committee on Big Data (TCBD). He is a recipient of 2021 IEEE Communications Society Outstanding Paper Award, 2021 IEEE Internet of Things Journal Best Paper Award, IEEE Computer Society 2020 Best Paper Award and IEEE Computer Society 2019 Best Paper Award Runner-Up. He is an IEEE Vehicular Technology Society Distinguished Lecturer. 
\end{IEEEbiography}

\begin{IEEEbiography}[{\includegraphics[width=0.95in,height=1.4in]{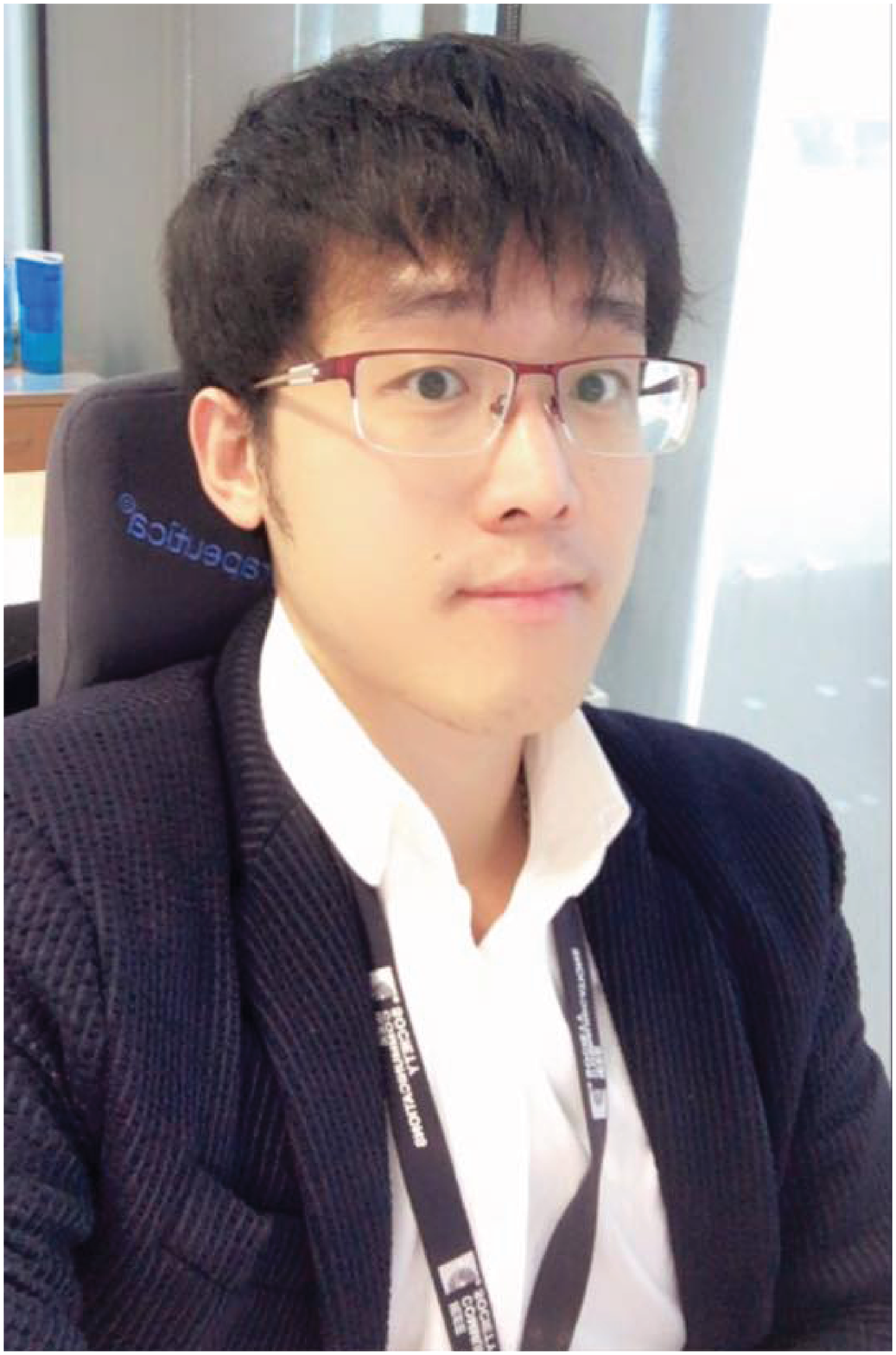}}]{Derrick Wing Kwan Ng} (Fellow, IEEE) received a bachelor's degree with first-class honors and a Master of Philosophy (M.Phil.) degree in electronic engineering from the Hong Kong University of Science and Technology (HKUST) in 2006 and 2008, respectively. He received his Ph.D. degree from the University of British Columbia (UBC) in Nov. 2012. He was a senior postdoctoral fellow at the Institute for Digital Communications, Friedrich-Alexander-University Erlangen-N\"urnberg (FAU), Germany. He is now working as a Scientia Associate Professor at the University of New South Wales, Sydney, Australia. His research interests include global optimization, physical layer security, IRS-assisted communication, UAV-assisted communication, wireless information and power transfer, and green (energy-efficient) wireless communications.
	
Dr. Ng has been listed as a Highly Cited Researcher by Clarivate Analytics (Web of Science) since 2018. He received the Australian Research Council (ARC) Discovery Early Career Researcher Award 2017, the IEEE Communications Society Leonard G. Abraham Prize 2023, the IEEE Communications Society Stephen O. Rice Prize 2022, the Best Paper Awards at the WCSP 2020, 2021, IEEE TCGCC Best Journal Paper Award 2018, INISCOM 2018, IEEE International Conference on Communications (ICC) 2018, 2021, 2023, IEEE International Conference on Computing, Networking and Communications (ICNC) 2016, IEEE Wireless Communications and Networking Conference (WCNC) 2012, the IEEE Global Telecommunication Conference (Globecom) 2011, 2021 and the IEEE Third International Conference on Communications and Networking in China 2008. He served as an editorial assistant to the Editor-in-Chief of the IEEE Transactions on Communications from Jan. 2012 to Dec. 2019. He is now serving as an editor for the IEEE Transactions on Communications and an Associate Editor-in-Chief for the IEEE Open Journal of the Communications Society.
\end{IEEEbiography}

\begin{IEEEbiography}[{\includegraphics[width=1in,height=1.25in,clip,keepaspectratio]{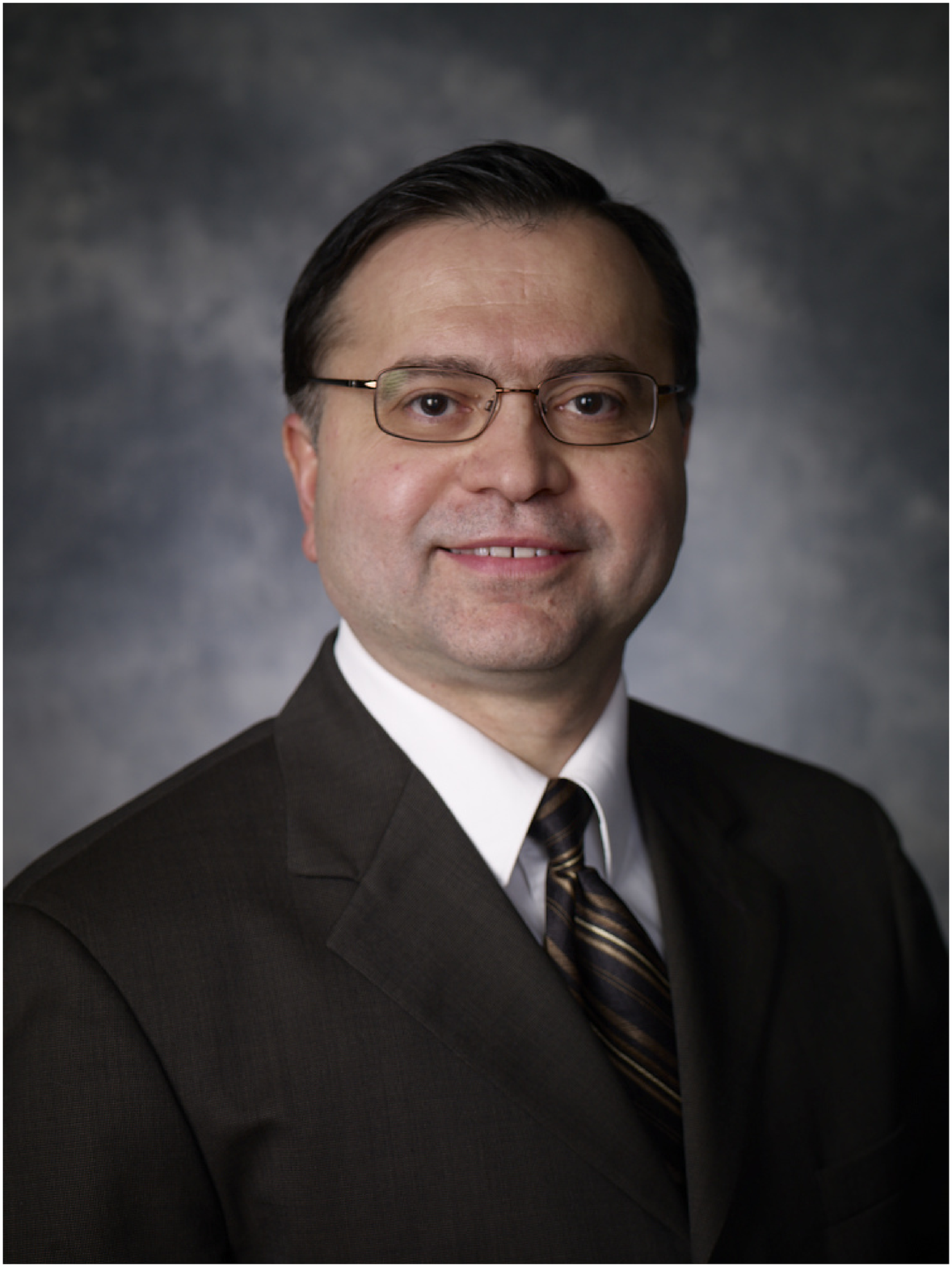}}]{Naofal Al-Dhahir} (Fellow, IEEE) is Erik Jonsson Distinguished Professor $\&$ ECE Associate Head at UT-Dallas. He earned his PhD degree from Stanford University and was a principal member of technical staff at GE Research Center and AT$\&$T Shannon Laboratory from 1994 to 2003.  He is co-inventor of 43 issued patents, co-author of over 550 papers and co-recipient of 5 IEEE best paper awards. He is an IEEE Fellow, AAIA Fellow, received 2019 IEEE COMSOC SPCC technical recognition award, 2021 Qualcomm faculty award, and 2022 IEEE COMSOC RCC technical recognition award. He served as Editor-in-Chief of IEEE Transactions on Communications from Jan. 2016 to Dec. 2019.  He is a Fellow of the US National Academy of Inventors and a Member of the European Academy of Sciences and Arts.
\end{IEEEbiography}

\end{document}